\documentclass{article}
\usepackage{graphicx}
\usepackage{amsfonts}
\usepackage{amsmath}
\usepackage{amssymb}
\usepackage{url}

\usepackage{dsfont}
\usepackage{fancyhdr}
\usepackage{indentfirst}
\usepackage{enumerate}
\usepackage[normalem]{ulem}
\usepackage{mathrsfs}
\usepackage{multirow}
\usepackage{diagbox}
\usepackage{subcaption}

 \usepackage[colorlinks=true,citecolor=blue]{hyperref}
\usepackage{amsthm}
\usepackage{color}

\definecolor{DarkGreen}{rgb}{0.2,0.6,0.2}

\definecolor{purple}{rgb}{0.6,0.3,0.8}

\usepackage{natbib}
\usepackage{comment}

\def\d{\mathrm{d}}

\newcommand{\E}{\mathbb{E}}

\newcommand{\R}{\mathbb{R}}

\newcommand{\N}{\mathbb{N}}

\newcommand{\dsquare}{\mathop{  \square} \displaylimits}

\newcommand{\p}{\mathbb{P}}
\newcommand{\X}{\mathcal{X}}

\newcommand{\cR}{\mathcal{R}}

\newcommand{\id}{\mathds{1}}

\renewcommand{\ge}{\geqslant}
\renewcommand{\le}{\leqslant}
\renewcommand{\geq}{\geqslant}
\renewcommand{\leq}{\leqslant}
\renewcommand{\epsilon}{\varepsilon}

\renewcommand{\cdots}{\dots}
\theoremstyle{plain}
\newtheorem{theorem}{Theorem}
\newtheorem{corollary}{Corollary}

\newtheorem{proposition}{Proposition}
\theoremstyle{definition}
\newtheorem{definition}{Definition}
\newtheorem{example}{Example}

\theoremstyle{remark}
\newtheorem{remark}{Remark}

\DeclareMathOperator*{\argmin}{arg\,min}

\newcommand{\VaR}{\mathrm{VaR}}

\newcommand{\ES}{\mathrm{ES}}

\usepackage{setspace}

\usepackage{footmisc}
\title{Quantiles under ambiguity  and risk sharing}

\author{Peng Liu\thanks{School of Mathematics, Statistics and Actuarial Science, University of Essex, UK (\href{mailto:peng.liu@essex.ac.uk}{peng.liu@essex.ac.uk}).}
\and 
Tiantian Mao\thanks{School of Management, University of Science and Technology of China, China (\href{mailto:tmao@ustc.edu.cn} {tmao@ustc.edu.cn}).}
\and
 Ruodu Wang\thanks{Department of Statistics and Actuarial Science,
  University of Waterloo,  Canada
(\href{mailto:wang@uwaterloo.ca}{wang@uwaterloo.ca}).}
}

\begin{document}

\maketitle

\begin{abstract}
Choquet capacities and integrals are central concepts in decision making under ambiguity or model uncertainty, pioneered by Schmeidler. 
Motivated by risk optimization problems  for quantiles under ambiguity, we study the subclass of  Choquet integrals, called Choquet quantiles, which generalizes the usual (probabilistic) quantiles, also known as Value-at-Risk in finance, from probabilities to capacities. 
Choquet quantiles share many features with probabilistic quantiles, in terms of axiomatic representation, optimization formulas, and risk sharing. We characterize   Choquet quantiles  via  only one axiom, called ordinality. We   prove that the inf-convolution of Choquet quantiles is again a Choquet quantile, leading to explicit optimal allocations in risk sharing problems for  quantile agents under ambiguity. A new class of risk measures,  Choquet Expected Shortfall, is introduced, which enjoys most properties of the coherent risk measure Expected Shortfall. Our theory is complemented by optimization algorithms, numerical examples, and a stylized illustration  with financial data.

\end{abstract}

\section{Introduction}

Choquet integrals are  a central concept in  decision making  under ambiguity, pioneered by \cite{S89}, and they are also fundamental tools for decision under risk, as in the celebrated theories of \cite{Y87} and \cite{TK92}.
We   study a special class of Choquet integrals, called Choquet quantiles, and their risk sharing problems. The Choquet quantiles are natural generalizations of the usual quantiles, with the interpretation of quantiles under ambiguity.  
The motivation for studying Choquet quantiles and their relevance in the risk sharing context can be illustrated by the following simple example in finance.

Suppose that a financial institution has an investment in   assets with future (after one period) aggregate random loss $X$, and the  institution has several subsidiaries, separated by their locations, business types, or other factors. 
Each subsidiary is allocated a random loss $X_i$ and 
is regulated by a risk measure, assumed to be a Value-at-Risk (VaR),\footnote{Although for evaluating market risk, VaR has been replaced by the Expected Shortfall (ES) in the recent Basel Accords(\cite{BASEL19}), it is still widely used for other types of risks, such as operational risk and credit risk.}  denoted by $\VaR_{\alpha_i}^\p$, which is the  $\alpha_i$-level   quantile of their financial loss distribution evaluated under an objectively specified probability measure $\p$.
Thus, $\VaR_{\alpha_i}^\p(X_i)$ is the risk exposure, also the capital requirement, of subsidiary $i$.
The total capital requirement (risk exposure)  of the financial institution,    is 
$$
\sum_{i=1}^n \VaR^\p_{\alpha_i}(X_i),
$$ 
and the institution can choose $(X_1,\dots,X_n)$ subject to $ X_1+\dots+X_n=X$ to minimize the total risk exposure. This is the risk sharing problem studied by \cite{ELW18}. 

In the above formulation, the probability $\p $ needs to be specified and unified across all subsidiaries. 
Arguably, an objectively truthful probability in  financial practice, even if it exists,  is never known for certain. The banks and insurance companies often do not know  or agree on the probability of a certain loss event. This concern is increasingly relevant in today's economy,  full of various kinds of uncertainty. 
In the Basel regulatory framework (\cite{BASEL19}), it is explicitly required that risks need to be quantified under many stressed probabilities.
This leads to ambiguity, also often called model uncertainty, in risk evaluation and decision making. 

To formulate risk evaluation under ambiguity, a common practice, instructed by \cite{BASEL19}, is to statistically specify  a collection of possible probabilities, and then use the worst-case risk measure. 
In our setting, this means to use
$\sup_{\mathbb Q\in \mathcal P_i} \VaR_{\alpha_i}^{\mathbb Q}(X_i)$ to evaluate the risk exposure for each subsidiary, where $\mathcal P_i$ 
is the set of possible probabilities, also called an ambiguity set or a set of priors. Using the worst-case risk measures under multiple probabilities is consistent with the   maxmin decision model of \cite{GS89} and it is a common approach in risk optimization under uncertainty. The total risk exposure becomes
\begin{align}\label{eq:main}
\sum_{i=1}^n\sup_{\mathbb Q\in \mathcal P_i} \VaR_{\alpha_i}^{\mathbb Q}(X_i).
\end{align}
The above problem can be formulated in a   more general context of risk sharing among $n$ agents, whose preferences are modelled by  worst-case quantiles. In such a general setting, an allocation   minimizes \eqref{eq:main} if and only if it is Pareto optimal (\cite{ELW18}). Quantile-based agents are also relevant outside finance, as axiomatizated by \cite{R10} in decision theory. 

The minimization of  \eqref{eq:main} subject to $X_1+\dots+X_n=X$ is one of the  main objectives of the paper.
The functional that maps $X$ to the minimum value of \eqref{eq:main}  is called an inf-convolution, and the corresponding minimizer is called an optimal allocation. 
To solve for the inf-convolution and optimal allocation requires us to 
understand better the worst-case quantiles as risk measures.

It turns out that the worst-case quantiles belong to a large class of mappings, which we call Choquet quantiles.
A Choquet quantile is a quantile with respect to a capacity, representing non-additive priors in decision theory,
and it belongs to the class of Choquet integrals. 
Choquet quantiles are not well explored in the literature, with the notable exception of \cite{C07} in the context of preference aggregation. The first part of our paper is  dedicated to a systemic study of Choquet quantiles.   Section  \ref{sec:def} contains the definitions of Choquet integrals, Choquet quantiles, and probabilistic quantiles, as well as their basic relations. Section \ref{sec:main} contains our main characterization result in Theorem \ref{th:main}, which shows that Choquet quantiles, as mappings $\cR$ from the set of bounded random variables to real numbers, can be characterized by only one axiom of ordinality, meaning  that $\cR(\phi(X))=\phi(\cR(X))$ for all random variables $X$ and  increasing and continuous real functions $\phi$. This result is similar to a characterization in \cite{C07} where two axioms are used to characterize Choquet quantiles.

Section \ref{sec:inf-convolution}   addresses the optimal risk sharing problem for Choquet quantiles, 
including those  for 
 worst-case quantiles under ambiguity in \eqref{eq:main}, as well as quantiles under heterogeneous beliefs, studied by \cite{ELMW20}. 
The main finding in Theorem \ref{thm:inf} is that the inf-convolution of Choquet quantiles is again a Choquet quantile, and its optimal allocation has an explicit formula. However, the inf-convolution of worst-case quantiles  
 is not necessarily a worst-case quantile. Hence, we  arrive at Choquet quantiles  even if we are  interested only in the subclass of worst-case quantiles or the subclass of probabilistic quantiles.
As a consequence of the obtained optimal allocation, 
using the worst-case quantiles to incorporate ambiguity cannot capture the tail risk, in a way similar to probabilistic quantiles, as discussed by \cite{ELW18}. 
 
 The Expected Shortfall (ES) is  another popular regulatory risk measure in banking and insurance,  as a counterpart to VaR.  
 To broaden the scope of our study, and  for a better understanding of Choquet quantiles, we 
 explore ES under ambiguity in Section \ref{sec:ES}. 
 The  Choquet  ES  is introduced, which is defined via integrating Choquet quantiles. 
  The Choquet  ES admits   three equivalent alternative formulations enjoyed by  ES: as a Choquet integral,  as the minimum of an optimization problem, and as a robust expectation.
 In particular, a generalization of the  formula of \cite{RU02} in Theorem \ref{th:min} connects Choquet ES and Choquet quantiles by a simple optimization problem.  
 Moreover, Choquet ES is a coherent risk measure for some classes of capacities. All those results suggest that Choquet  ES is a natural generalization of ES to incorporate ambiguity, maintaining the main features of ES.  The inf-convolution of Choquet $\ES$ is also studied. 
 As a notable remark, the supremum of ES over a set of probabilities generally cannot be written as a Choquet ES, in sharp contrast to the case of Choquet quantiles.
 
  After the theoretical development in Sections \ref{sec:main}--\ref{sec:ES},    we turn to solving the risk sharing problem for worst-case quantiles under ambiguity in \eqref{eq:main}  by introducing some numerical algorithms in Section \ref{sec:numerical}.  
  Finally, Section \ref{sec:realdata} presents a stylized financial data analysis for sharing financial losses, where we use some standard approach to build the ambiguity sets used by agents. 
  Section \ref{sec:conclusion} concludes the paper. All proofs are relegated to Appendix \ref{app:A}, and  some additional numerical details and figures are provided in Appendix \ref{app:numerical}.

 \subsection{Related literature}
 \label{sec:11}

Quantiles, called VaR in finance,   are one of the dominating regulatory risk measures to quantify the regulatory capital; see  \cite{MFE15} for a general treatment. In decision theory, the decision maker uses quantiles to rank the interested prospects, and  preferences induced by quantiles are studied by \cite{R10}, \cite{CG19, CG22} and \cite{CGNQ22}.  
Quantiles work as either objectives or constraints in the  portfolio optimization problem; see e.g., \cite{BS01} and \cite{HZ11}. 
Conditional quantiles are used to quantify systemic risks by \cite{AB16}.

Ambiguity is an active area of research in decision theory. 
Modeling decision under ambiguity by a worst-case risk measure or utility over probabilities is a classic approach, initially axiomatized by \cite{GS89} (maxmin expected utilities) and further explored and extended by \cite{HS01} (multiplier preferences), \cite{GMM04} (the $\alpha$-maxmin model), 
\cite{MMR06} (variational preferences), 
and \cite{CHMM21} (preferences under model misspecification). 
Such a formulation is also central to the theory of coherent risk measures of 
\cite{ADEH99}.
The worst-case risk measures under multiple probabilities are widely studied in the field of risk optimization under uncertainty; see 
 \cite{GOO03}, \cite{NPS08}, and \cite{ZF09} for the cases of VaR and ES.

This definition of Choquet quantiles, as a special class of Choquet integrals, falls in the realm of decision making with non-additive priors, as in \cite{Y87} (dual utility thoery), \cite{Qui82} (rank-dependent expected utility), and \cite{S89} (Choquet expected utility).
 \cite{MM04} offered a summary of capacities and Choquet integrals in decision making. 
 Choquet quantiles are closlely related to unanimity games in game theory (e.g., \cite{M95}). As mentioned above, the closest work on Choquet quantiles to our study is 
 \cite{C07}, although the motivation and interpretation there are quite different from ours, 
 as we treat  
Choquet quantiles as monetary risk measures  in the sense of \cite{FS16} and are mainly motivated by applications in risk management.

Risk sharing and inf-convolution for monetary risk measures have been studied by \cite{BE05}, \cite{JST08},  \cite{FS08}, and \cite{D12},   assuming convexity. Note that Choquet quantiles are not convex in general. Without convexity, risk sharing problems  for risk measures under ambiguity or heterogeneous probabilities among agents  are more sophisticated, and studied recently by
  \cite{ELMW20}, \cite{L24} and \cite{LMWW22, LTW24}.  
The issue of VaR being unable to capture tail risk is related to the non-robustness of VaR in risk optimization under ambiguity, as discussed by  \cite{ESW22}.

ES has replaced VaR as the standard risk measure for market risk in the recent Basel Accords because it better captures the tail risk (\cite{BASEL16, BASEL19}) and  diversification effect; see \cite{RU02} for optimization of ES, also called Conditional VaR (CVaR), and \cite{WZ21b} for an axiomatization of ES in the context of Basel Accords.
With the current regulatory frameworks, the two classes  of risk measures, VaR and ES, continue to coexist in financial and insurance regulation.  
\cite{EPRW14} and \cite{EKT15} discussed various issues of choosing risk measures in financial regulation.

\section{Choquet quantiles}\label{sec:def}
Let $\X$ be the set of bounded functions (also called random variables) on a measurable space $(\Omega,\mathcal F)$.
Throughout, terms like ``increasing" and ``decreasing" are in the non-strict sense. 
\begin{definition}
\label{def:0}
\begin{enumerate}[(i)]
\item {A function  $w:\mathcal F\to \R$ is \emph{increasing} if $w(A)\le w(B)$ whenever $A\subseteq B$.} 
\item 
    A \emph{capacity} is an increasing function $w:\mathcal F\to \R$
satisfying $w(\varnothing)=0$ and $w(\Omega)=1$. 
\item 
A capacity is \emph{binary} if it takes values in $\{0,1\}$.
\item A \emph{probability} is a capacity further satisfying countable additivity.
\item A \emph {sup-probability} 
is a capacity $\sup_{\mathbb Q\in \mathcal P} \mathbb Q$  for a set of probabilities $\mathcal P$.
\end{enumerate}
\end{definition}
All capacities  and probabilities are   defined on $(\Omega,\mathcal F)$ without explicitly mentioning it.   
A sup-probability is necessarily subadditive, that is $w(A\cup B)\le w(A)+w(B)$ for $A,B\in \mathcal F$.
For two constants $a,b$, we write $a \wedge b$ for their maximum  and $a \vee b$ for their minimum.
For two capacities $w,w'$,
their minimum $w\wedge w'$
and their maximum $w\vee w'$
are defined point-wise, and they are capacities. 
Similarly, the infimum $\bigwedge_{w\in W} w=\inf_{w\in W}w$ and  supremum $\bigvee_{w\in W} w=\sup_{w\in W}w $ over a set  $W$ of capacities are understood point-wise; this applies to Definition \ref{def:0} (v).


For $X\in \X$ and  a capacity   $w:\mathcal F\to \R$,  the \emph{Choquet integral},
 denoted by $ \int X \d w$, is defined as\footnote{In the definition of $I_w$,  using $w(X\ge x)$  is equivalent to  using $w(X>x)$; see Proposition 4.8 of \cite{MM04}. This also applies to \eqref{eq:rep1} and \eqref{eq:rep2}.}
$$\int X \d w= \int_{-\infty}^0  \left( w(X\ge x) - 1 \right)\d x + \int_{0}^\infty w (X\ge x) \d x.
$$
We also denote by  $I_w:\X\to \R$ the mapping  $$I_{w}(X)=\int X \d w,~~~~X\in \X.$$
Two random variables $X$ and $Y$ are \emph{comonotonic} if $(X(\omega)-X(\omega'))(Y(\omega)-Y(\omega'))\ge 0$ for all $\omega,\omega'\in\Omega$. 
Choquet integrals satisfy the crucial property of comonotonic additivity, i.e., $I_w(X+Y)=I_w(X)+I_w(Y)$ if $X$ and $Y$ are comonotonic, which is one of the characterizing property of Choquet integrals (\cite{S86}).

We now introduce the core concept of the paper, the Choquet quantiles.

\begin{definition}\label{def:main}
\begin{enumerate}[(i)]
    \item A \emph{Choquet quantile} is the mapping $  I_v$  for some binary capacity $v$.
    \item The \emph{left $\alpha$-Choquet quantile for a capacity $w$}  and $\alpha \in (0,1]$ is defined by
\begin{equation}
\label{eq:rep1} 
Q^w_\alpha(X)=   \inf \{ x\in \R : w(X\ge x) \le 1-\alpha\},~~X\in \X.
\end{equation}
    \item  The  \emph{right $\alpha$-Choquet quantile for a capacity $w$}  and $\alpha \in [0,1)$ is defined by
\begin{equation}
\label{eq:rep2} \overline{Q}^w_\alpha(X)=   \inf \{ x\in \R : w(X\ge x) < 1-\alpha\},~~X\in \X.\end{equation}
\end{enumerate} 
\end{definition}
If $w=\p$ is a probability, then $Q^w_\alpha$ or $\overline{Q}^w_\alpha$ is called a \emph{probabilistic quantile}, which is usually how quantiles are defined in statistics and probability.

It may not be obvious from the definition, but the three classes (i)--(iii) are equivalent, as observed by \cite{C07}, which we formally state in the following proposition.

\begin{proposition}
\label{prop:basic-def}
 The three classes (i)--(iii) in Definition \ref{def:main} are equivalent. 
\end{proposition}


The equivalence in Proposition \ref{prop:basic-def} can be explained by the following correspondence between $(w,\alpha)$ and $v$. 
For $I_v$ with a binary capacity $v$, define $w=v$ and $\alpha =1/2$, and then we have
$$
 I_v = Q^w_\alpha = \overline{Q}^w_\alpha.
$$
For  $Q^w_\alpha$ with $\alpha\in (0,1]$ (resp. $\overline{Q}^w_\alpha$ with $\alpha\in [0,1)$) and a capacity $w$, define a binary capacity $v$  as $v(A)=\id_{\{w(A) >1-\alpha\}}$ (resp. $v(A)=\id_{\{w(A) \ge 1-\alpha\}}$), $A\in \mathcal F$, and then we have
$$
Q^w_\alpha = I_{v} ~~({\rm resp.} ~  \overline{Q}^w_\alpha= I_{v}).
$$ 

For a binary capacity $v$, we define the {\em null set} of $v$, denoted by  $\mathcal C_v$, as 
$$
\mathcal C_v=\{A\in \mathcal F: v(A)=0\}.
$$
A binary capacity  $v$ is determined by its null set.  The Choquet integral $I_v$ can be obtained from $\mathcal C_v$ as follows
$$
 I_v(X)=\inf \{x\in \R: \{X\ge x\}\in \mathcal C_v\}=\inf \{x\in \R: \{X> x\}\in \mathcal C_v\},~~X\in\X.
$$
Moreover, one can easily check that for binary capacities $v_1$ and $v_2$, it holds that $\mathcal C_{v_1\vee v_2} = \mathcal C_{v_1}\cap \mathcal C_{v_2}$ and $\mathcal C_{v_1\wedge v_2} = \mathcal C_{v_1}\cup \mathcal C_{v_2}$.
In particular, if $I_v$ is a probabilistic quantile, that is, $I_v=Q^\p_\alpha$ for some probability $\p$, then 
$$
\mathcal C_v=\{A\in \mathcal F: \p(A)\le 1- \alpha\}. 
$$

\section{Characterization and properties}\label{sec:main}

In this section, we   give an axiomatic characterization of Choquet quantiles based on a   property of ordinality,
and   study several properties of Choquet quantiles that are useful for risk sharing problems with ambiguity. 

\subsection{Axiomatic characterization} 

We first introduce a key property, ordinality.
We say that a mapping $\cR$ satisfies \emph{ordinality} if  $\cR\circ \phi =\phi\circ\cR$ for all  increasing and continuous functions $\phi:\R\to \R$.
Here, $\cR\circ \phi$ is understood as $X\mapsto \cR( \phi \circ X)$.

To make the   characterization result more general and applicable in different fields, 
we consider an algebra $\mathcal A\subseteq \mathcal F$, and 
$\cR$ as a mapping from $ \X_{\mathcal A}$, where $ \X_{\mathcal A}$ is the set of bounded functions on 
$\Omega$ measurable with respect to an algebra $\mathcal A$.\footnote{Recall that  $X$ is measurable with respect to an algebra $\mathcal A$  if $\{X>x\}\in \mathcal A$ and $\{X\geq x\}\in \mathcal A$ for all $x\in\R$.}
It is common to work with algebras instead of $\sigma$-algebras in the study of  Choquet integrals in game theory; see e.g.,  \cite{MM04}.
Since we mainly interpret functions in $\X$ as random variables, we stick to $\sigma$-algebra $\mathcal F$ in other parts of the paper. 
 The following result   also holds for $\cR:\X\to \R$ because a $\sigma$-algebra is an algebra.

\begin{theorem}
\label{th:main}
Let $\mathcal A$ be an algebra on $\Omega$. 
For    $\cR:\mathcal  X_{\mathcal A}\to \R$, the following are equivalent:
\begin{enumerate}[(i)]
\item $\cR$ satisfies ordinality;
\item  $\cR$ is a Choquet quantile.
\end{enumerate}
\end{theorem}

We briefly sketch the main idea of the proof of Theorem \ref{th:main}.
The direction (ii)$\Rightarrow$(i) can be shown by straightforward arguments. 
To prove the direction (i)$\Rightarrow$(ii), we first check three facts: (a) $\cR(X)$ takes values only in the closure of $\{X(\omega): \omega \in \Omega\}$;  (b) $\cR$ is comonotonic additive; (c) $\cR$ is positively homogeneous (defined in Section \ref{sec:property}). 
Fact (a) follows from ordinality. Fact (b) follows from (a), ordinality, and applying Denneberg's lemma (Proposition 4.5 of \cite{D94}).
Fact (c) is a direct consequence of ordinality. 
Next, we consider the set $\mathcal X_0$ of all simple functions in $\mathcal X_A$, i.e., those which take finitely many values. By Proposition 1 of \cite{S86},
facts
(b) and (c) imply that $\cR$ is a Choquet integral $\cR(X)=\int X\d v$ on $\mathcal X_0$ for some $v:\mathcal F \to \R$ with $v(\varnothing)=0$.
Then  fact (a)  implies that $v$ can only take values in $\{0,1\}$ and $v(\Omega)=1$. Using (b) and (c) together with some simple manipulations shows that $v$ is increasing, and hence $v$ is a capacity. 
This shows that $\cR$ is a Choquet quantile on $\mathcal X_0$. The remaining task  is to show that the same representation holds on $\mathcal X_A$, and this is achieved by using some continuity arguments.

As discussed by \cite{C07}, the interpretation of ordinality in decision making is that a possibly increasing and continuous scale change on measuring random outcomes does not affect their relative desirability.
Putting this into a decision-theoretic framework, let the preference $\preceq$ be a total preorder on $\X_{\mathcal A}$\footnote{A total preorder is a binary relation on $\X_{\mathcal A}$ satisfying, for all $X,Y,Z\in \X_{\mathcal A}$, (i) $X\preceq Y$ and $Y\preceq Z$ imply $X\preceq Z$, and (ii) $X\preceq Y$ or $Y\preceq X$ holds.}. Its equivalence relation and strict relation are denoted by $\sim$ and $\prec $, respectively.
The natural formulation of  ordinality in this context is
\begin{align}\label{eq:preference0}
X\preceq Y ~\Longrightarrow~ \phi(X) \preceq \phi(Y)
\end{align}
for all increasing and continuous functions $\phi$.   We say the preference $\preceq$ is \emph{increasing} if $0\prec 1$.
If  $\preceq$ can be \emph{numerically represented}, that is, there exists $\mathcal R: \X_{\mathcal A}\to\R$ such that $X\preceq Y$ if and only if $\mathcal R(X)\le \mathcal R(Y)$, then the property \eqref{eq:preference0} is equivalent to the existence of some numerical representation $\mathcal R$ satisfying ordinality.
\begin{proposition}  \label{prop:4} Let the preference $\preceq$ be an increasing total preorder on $\X_{\mathcal A}$ and suppose it can be numerically represented. Then the following two statements are equivalent.
\begin{enumerate}
\item[(i)] If $X\preceq Y,$ then $  \phi(X) \preceq \phi(Y)$ holds  for all increasing and continuous functions $\phi$;
\item[(ii)] There exists a numerical representation $\mathcal R$ satisfying ordinality.
\end{enumerate}
\end{proposition}
 \cite{C07} characterized Choquet quantiles using two axioms: ordinality (only for strictly increasing and continuous functions) and monotonicity. For  our result in  Theorem 1, we show that only one axiom, ordinality, is sufficient, which implies a non-trivial fact that monotonicity is implied by ordinality already.
Although ordinality requires $\cR\circ \phi =\phi\circ\cR$ to hold for all increasing continuous $\phi$, 
for the characterization in Theorem \ref{th:main}, 
it suffices to only require it to hold for two-piece linear functions, assuming that $\cR$ is continuous with respect to the sup-norm. 
This is because compositions of two-piece linear functions can uniformly approximate any continuous increasing function on a bounded interval.
 
\begin{remark}  
\cite{FLW22} showed that
 $\cR$ is   law-based with respect to $\p$ (i.e., $\cR(X)=\cR(Y)$ for all $X$ and $Y$ that are identically distributed under $\p$) and ordinal if and only if
   $\cR $ is a  probabilistic quantile for $\p$.
  Our Theorem \ref{th:main}  is  more general than that result in two aspects. First, Theorem \ref{th:main}  requires no assumption on the law-basedness of the functional, whereas this is an essential assumption in \cite{FLW22}, resulting in completely different techniques for the proofs. Second, Theorem \ref{th:main}  is valid for more general measurable spaces $(\Omega,\mathcal A)$ with $\mathcal A$ being an algebra, whereas the other result requires a $\sigma$-algebra. By imposing the  law-basedness in Theorem \ref{th:main}, we immediately  recover the result on probabilistic quantile. 
\end{remark}

\subsection{Properties}\label{sec:property}

This section provides some properties of Choquet quantiles that are useful in studying risk sharing problems with Choquet quantiles. 
For a class  $W$ of capacities and a class $V$ of binary capacities, we write 
\begin{equation}
\label{eq:max-w} \overline w =\sup_{w\in W} w;~~\underline{w} =\inf_{w\in W} w;~~ \overline{v} =\inf_{v\in V} v~~{\rm and }~~\underline{v} =\inf_{v\in V} v.
\end{equation}
We first give the closure property of the Choquet quantiles  under the supermum and the infimum of the capacities.
\begin{proposition}\label{pro:supclose}
Let $W$ be a set of capacities,  $V$ be a set of binary capacities, and $\alpha \in(0,1)$.
Using the notation in \eqref{eq:max-w}, we have
 $$\sup_{w\in W}   Q^w_\alpha = Q^{\overline w }_\alpha
 \mbox{~~~and~~~}
  \sup_{v\in V}   I_v = I_{{\overline v}};$$
  $$\inf_{w\in W}   \overline{Q}^w_\alpha = \overline{Q}^{\underline{w} }_\alpha
 \mbox{~~~and~~~}
  \inf_{v\in V}   I_v = I_{\underline{v}}.$$
\end{proposition}


In Proposition \ref{pro:supclose}, we note that the supremum is associated with $Q^w$ and the infimum is associated with $\overline{Q}^w$.
This distinction is intentional. 
Below, we give an example to show that 
$\sup_{w\in W}   \overline{Q}^w_\alpha = \overline{Q}^{\overline w }_\alpha$ does not hold in general. The case of $\inf_{w\in W}   Q^w_\alpha = Q^{\underline w }_\alpha$ is similar. 

\begin{example} Let $(\Omega,\mathcal F) = ([0,1], \mathcal B([0,1]))$, $X(x)=\id_{\{x\ge 1/2\}}$ for $x\in \Omega$, $W=\{w_n, n\in\N\}$ satisfying $w_n([1/2,1])=1/2-1/n$ for $n\in\N$, and $\alpha=1/2$. It holds that  $w_n(X\ge 1) = w_n([1/2,1]) = 1/2-1/n< 1 -\alpha$, and thus, 
$ 
\overline{Q}^{w_n}_\alpha(X) = 0
$ for $n\in\N.$ This implies $
\sup_{w\in W}\overline{Q}^{w}_\alpha(X) = 0.
$
One can verify that $\overline w :=\sup_{w\in W} w$ satisfies $\overline{w} ([1/2,1])=1/2 \not <1-\alpha$ and thus, $
\overline{Q}^{\overline{w} }_\alpha(X) = 1.
$
\end{example}

Note that  the supermum of  sup-probabilities is still a sup-probability. 
Therefore, the supremum of  quantiles for sup-probabilities is 
again a quantile for a sup-probability  by   Proposition~\ref{pro:supclose}. 

In particular, Proposition~\ref{pro:supclose} implies, for any set $\mathcal P$ of probabilities and $\alpha \in (0,1)$, 
$$
\sup_{\mathbb P\in \mathcal P}Q^{\mathbb P}_\alpha = Q^{\overline{\mathbb P}}_\alpha,
$$
where $\overline{\mathbb P}= \sup_{\mathbb P\in \mathcal P} \mathbb P$. 
In other words,  the worst-case quantiles mentioned in the Introduction are   quantiles for sup-probabilities,
and hence they are also Choquet quantiles. 




We briefly discuss Choquet quantiles as risk measures. Common properties of risk measures $\cR:\X\to \R$ include
\begin{enumerate}[(a)]
\item Monotonicity: $\cR(X)\le \cR(Y)$ whenever $X\le Y$;
\item Translation invariance: $\cR(X+c)= \cR(X)+c$ for all $X\in \X$ and $c\in\R$;
\item Positive homogeneity: $\cR(cX)= c\cR(X)$ for all $X\in \X$ and $c>0$;
\item Convexity: $\cR(\lambda X+(1-\lambda)Y)\leq\lambda \cR(X)+(1-\lambda)\cR(Y)$ for all $\lambda\in [0,1]$ and $X,Y\in\X$.
\end{enumerate}
A mapping $\cR:\X\to\R$ is called a coherent risk measure (\cite{ADEH99}) if it satisfies 
the above four properties, whereas it is called a  monetary risk measure if it satisfies the first two properties. 
It is well known (e.g.,  Proposition 4.11 of \cite{MM04}), and straightforward to check, that the first three properties are always satisfied by all Choquet integrals, including Choquet quantile.  

We next discuss convexity, which is important for optimization problems.  
A capacity $w$ is \emph{submodular} if $w(A\cup B)+w(A\cap B)\leq w(A)+w(B)$ for all $A,B\in \mathcal F$. A Choquet integral $I_w$ is convex if and only if $w$ is submodular (Theorem 4.6 of \cite{MM04}). 
For a binary capacity,  submodularity is quite restrictive. 
The general message is that convexity (thus coherence) is rarely satisfied by Choquet quantiles. 
Recall that   a probabilistic quantile is never convex unless it is an essential supremum.

\begin{proposition} \label{prop:7} For a binary capacity $v$, 
  a Choquet  quantile $I_v$ is convex
if and only if $\mathcal C_{v}$
is closed under union.
\end{proposition}

A main example of convex Choquet quantiles $I_v$ is given below.
Let $\mathcal C_v=\{A\in\mathcal F: A\subseteq B\}$ for some $B\in\mathcal F$ with $B\neq \Omega$. Obviously, $\mathcal C_v$ is closed under union. It follows that $$I_v(X)=\inf\{x\in\R: \{X\geq x\}\subseteq B\} = \sup_{\omega\in \Omega\setminus B} X(\omega).$$
In particular, if $B=\varnothing$, then $I_v(X)=\sup X$ for all $X\in\mathcal X$.

Next we show that any convex Choquet quantile can be written as the above form in the commonly used measurable space.
We say a capacity $w$ is \emph{continuous} if $w(A_n)\to 0$ as $A_n\downarrow \varnothing$ as $n\to\infty$. Next, we will show that for continuous $v$, convex $I_v$ has a very special form. We denote  by $\mathcal B([0,1])$ the set of all Borel sets on $[0,1]$ and  by $[n]$ the set $\{1,\dots, n\}$.
\begin{proposition}  \label{prop:8}
  Let $(\Omega, \mathcal F)=([n], 2^{[n]})$, $(\mathbb N, 2^{\mathbb N})$, or $([0,1], \mathcal B([0,1]))$ and suppose  $v$ is a continuous binary capacity. Then  $I_v$ is convex if and only if there exists a finite set $A\subseteq \Omega$   such that $I_v(X)=\max_{\omega\in A} X(\omega)$ for all $X\in \X$.
\end{proposition}

Proposition \ref{prop:8} suggests that convexity and continuity, both useful properties in optimization, are almost conflicting for Choquet quantiles, as they lead to very restrictive forms of Choquet quantiles. 
Therefore, in the risk sharing problems that we study in this paper, generally, no convexity can be expected, and specific analysis is needed.

 The next result shows that the Choquet quantile   for a capacity has the semi-continuity as the probabilistic quantile. 
\begin{proposition} \label{prop:2} For $X\in\X$ and a capacity $w$,  $\alpha\mapsto Q^w_\alpha(X)$ is left-continuous on $(0,1]$ and $\alpha\mapsto \overline{Q}^w_\alpha(X)$ is right-continuous on $[0,1)$.
\end{proposition}

The above result implies  that the value of $Q^w_\alpha(X)$ at $\alpha=1$  and that of $\overline{Q}^w_\alpha(X)$ at $\alpha=0$
are determined by their values on $(0,1)$, and hence we only need to consider the domain $(0,1)$ for the Choquet quantile functions, as for probabilistic quantile functions.

Finally, we make a simple observation that,   any Choquet integral can be represented as an integral of the Choquet quantiles, similarly to the fact that an expectation can be written as the integral over probabilistic quantiles. 
Formally, for $X\in\mathcal X$ and a capacity $w$, we have
   $$
\int X \d w = \int_0^1 Q^w_\alpha(X) \d \alpha =  \int_0^1 \overline{Q}^w_\alpha(X) \d \alpha.
$$
The first equality is given by Proposition 4.86 of \cite{FS16} and the second one follows from that the discontinuity points of a monotone function is at most countable, and thus, the points at which $ Q^w_\alpha(X) $ and   $\overline{Q}^w_\alpha(X) $  are not equal are at most countable.

\section{Risk sharing with Choquet quantiles}\label{sec:inf-convolution}

We now turn to the problem of risk sharing with $n$ agents  using  Choquet quantiles as their preference functionals to minimize. For instance, these agents may represent the subsidiaries in the Introduction, but they can also represent several participants in a risk sharing pool. 

\subsection{Setting}

We interpret a random variable $X$ as the total risk (with positively values representing losses) in a risk-exchange economy. 
Given    $X\in \mathcal{X}$, we define the set of \emph{allocations} of $X$ as
\begin{align}
\mathbb{A}_n(X)=\left\{(X_1,\ldots,X_n)\in \mathcal{X}^n: \sum_{i=1}^nX_i=X\right\}. \label{eq:intro1}
\end{align}
In an allocation $(X_1,\dots,X_n)$, $X_i$ is interpreted as the amount of risk allocated to agent $i$.
The  \emph{inf-convolution} of   functionals $\rho_1,\dots,\rho_n:\X\to \R$ is defined as
\begin{equation} \label{eq-Paretogeneral}
 \dsquare_{i=1}^n \rho_i (X)  = \inf\left\{\sum_{i=1}^n \rho_i(X_i): (X_1,\cdots,X_n)\in \mathbb A_n(X)  \right\},~~X\in \X.
\end{equation}
It is well-known that for agents modelled by {monetary risk measures}, including Choquet integrals, 
Pareto optimality is equivalent to minimality of the total risk, that is, solving \eqref{eq-Paretogeneral}; see Proposition 1 of \cite{ELW18}.
We will call such allocations as \emph{optimal allocations} of $X$.

For a capacity $w$,   we write
$$
\VaR^w_{\alpha}= Q^w_{1-\alpha}\mbox{~~~and~~~}\overline{\VaR}^w_{\alpha}= \overline{Q}^w_{1-\alpha}.
$$
This choice of notation helps us to explain some results in the literature in a concise form.
As shown by \citet[Corollary 2]{ELW18}, 
for an atomless probability measure $\p$, the inf-convolution of probabilistic quantiles has the following explicit formula
\begin{equation}
    \label{eq:ELW18}
   \dsquare_{i=1}^n\VaR_{\alpha_i}^{\p}(X)=\VaR_{\sum_{i=1}^n\alpha_i}^{\p}(X), \mbox{~~~for all $X\in \X$,}
\end{equation} 
where $\VaR_\alpha$ is defined as $-\infty$ for $\alpha\ge 1.$ 
Risk sharing for probabilistic quantiles with respect to different probabilities  is studied by \cite{ELMW20}, who obtained an implicit formula of the inf-convolution. 
We will show below that Choquet quantiles admit a similar result; in particular, 
 the inf-convolution of Choquet quantiles is again a Choquet quantile or $-\infty$. 
\subsection{Inf-convolution and optimal allocations}
In what follows, we say that a set $\mathcal C\subseteq \mathcal F$ is closed from below if  for any increasing chain\footnote{This means $A_j\subseteq A_{j+1}$ for all $j\in \N$.} $(A_j)_{j\in \N}$ in $\mathcal C$, we have  $\bigcup_{j\in \N}A_j \in\mathcal C$. Let $\Pi_n$ be the set of all measurable $n$-compositions of $\Omega$,\footnote{A composition is a partition with order.} that is, 
$$
\Pi_n=\left\{(A_1,\dots,A_n)\in \mathcal F^n: \bigcup_{i=1}^n A_i =\Omega\right\}.
$$
 We are now ready to present the inf-convolution and optimal allocations for Choquet quantiles. 
\begin{theorem}\label{thm:inf}
  Let $v_1,\dots,v_n$ be binary capacities.  The following conclusions hold.
  \begin{enumerate}
  \item[(i)] We have
$$
\dsquare_{i=1}^n I_{v_i}= I_v,
$$
where $v$ is a binary capacity   with null set given by
 \begin{align}\label{eq:vnull-2}\mathcal C_v= \left\{\bigcup_{i=1}^n  A_i: A_i\in \mathcal C_{v_i},~i\in [n]\right\}.
 \end{align}
 \item[(ii)] If $x^*=I_v(X)>-\infty$ and  $\mathcal C_v$ is closed from below, then an optimal allocation of $X$ exists and it is given by
    \begin{align} \label{eq:optimalsolution} X_i^*=(X-x^*)\id_{A_i^*}+\frac{x^*}{n},~i\in [n],
    \end{align}
    for some $(A_1^*,\dots, A_n^*)\in \Pi_n$ satisfying $A_i^*\cap \{X>x^*\}\in \mathcal C_{v_i}$ for $i\in [n]$. 

      \end{enumerate}
\end{theorem}

The next result shows that the inf-convolution and optimal allocation  in Theorem \ref{thm:inf} enjoy an invariance property under continuous and increasing 
 transformations, thanks to the ordinality of Choquet quantiles in Theorem \ref{th:main}.
\begin{proposition}
\label{prop:invariance}
Fix $X\in \X$ and use the same notation in Theorem \ref{thm:inf}.
 For  any continuous and increasing $\phi:\R\to\R$, it holds that  $$
\dsquare_{i=1}^n I_{v_i} ( \phi(X)) =\phi( I_v (X))=\phi\left(\dsquare_{i=1}^n I_{v_i} ( X)\right).
$$
Moreover,
 an optimal allocation of $\phi(X)$ is given by $(Y_1^*,\dots,Y^*_n)$, where 
$$
Y_i^*=(\phi(X)-\phi(x^*))\id_{A_i^*}+\frac{\phi(x^*)}{n},~~~~i\in [n].
$$   
\end{proposition}

Proposition \ref{prop:invariance} implies that if an optimal allocation for $X$ has been solved, 
then we can immediately obtain an optimal allocation for $\phi(X)$, because the composition $(A_1^*,\dots,A_n^*)$ and the value $x^*$ are identical across the two problems. This observation is convenient in the numerical examples in Section \ref{sec:numerical}.

Note that the continuity  of $\mathcal C_v$
in  part (ii) of Theorem \ref{thm:inf} cannot be removed, because one may have $\{X>x^*+\epsilon\}\in \mathcal C_v$ for all $\epsilon>0$ but $\{X>x^*\}\not \in \mathcal C_v$. In this case, an optimal allocation may not exist, because $(A_1^*,\dots,A_n^*)$ in \eqref{eq:optimalsolution} cannot be constructed. 

The result in part (i) of Theorem \ref{thm:inf} can be equivalently stated for $Q^{w_i}_{\alpha_i}$
instead of $I_{v_i}$. For
  capacities $w_1,\dots,w_n$ and $\alpha_1,\ldots,\alpha_n\in (0,1)$, 
  let $v_i=\id_{\{w_i>1-\alpha_i\}}$, $i\in [n]$. 
Using the equivalence between the two classes of $Q_\alpha^w $ and $I_v$, we immediately get
$$
\dsquare_{i=1}^n Q^{w_i}_{\alpha_i} = \dsquare_{i=1}^n I_{v_i}= I_v = Q^{v}_{1/2},
$$
 where $v$ is a binary capacity  with null set given by \eqref{eq:vnull-2}.  

\subsection{Discussions}

From Theorem \ref{thm:inf}, we can see that the tail risk of $X$ on the event $\{X>x^*\}$
is shared by all agents, but the severity of the loss $X$ is ignored as it does not contribute to the calculation of $I_v(X)$.
Therefore, using Choquet quantiles, including VaR under ambiguity,   ignores the tail risk in the risk sharing context, in a similar way to   VaR, as discussed by \cite{ELW18}. 
Incorporating ambiguity into the risk measures can only change the level $x^*$ (beyond which the tail risk is ignored), but not addressing the issue of ``not capturing tail risk'', which is the main drawback of VaR specified by \cite{BASEL16}.
Therefore, although incorporating ambiguity can provide a more conservative approach to risk management, its robustness properties are similar to VaR. 

In \eqref{eq:optimalsolution}, the specification of $A^*$ on $\{X\le x^*\}$
is arbitrary. For instance, we can take $A_1^*\supseteq \{X\le x^*\}$ and  $A_i^*\cap\{X\le x^*\}=\varnothing$ for all $i>1$. 

The value $x^*$ in \eqref{eq:optimalsolution} can be changed to $y<x^*$, leading to 
  \begin{align} \label{eq:optimalsolution2} X_i^*=(X-y)\id_{A_i^*}+\frac{y}{n},~~~~i\in [n].
    \end{align}
    We can verify that  \eqref{eq:optimalsolution2}  also gives an optimal allocation for the risk sharing problem. If $X$ is bounded and $y$ is small enough, then $X>y$, and we obtain a \emph{counter-monotonic} allocation (that is, each pair $(X_i^*,-X_j^*)$ is comonotonic for $i\ne j$).
Such allocations typically 
represent a gambling behaviour, and they are optimal for certain risk-seeking agents, as studied by \cite{LLW24}.
Therefore, the behaviour of agents using Choquet quantiles in risk sharing problems is similar to risk-seeking agents in the above sense. 

If we interpret the risk measures $\rho_1,\dots,\rho_n$ as capital requirement (which is the main interpretation of VaR) for each agent,
then $\dsquare_{i=1}^n\rho_i(X)$ represents the least amount of total capital required by jointly holding the loss $X$ among the agents.
Using Choquet quantiles as the risk measures,   the total capital requirement can typically be reduced significantly (e.g., 50\% for some common loss distributions), as seen from the numerical results in Section \ref{sec:numerical} and the real-data example in Section \ref{sec:realdata}.

\subsection{Some special cases}

We now discuss a few special cases of risk sharing with Choquet quantiles. 
When the capacities are sup-probabilities, 
the inf-convolution has a simple form, which follows by combining Proposition \ref{pro:supclose}  and Theorem~\ref{thm:inf}.
\begin{corollary}\label{co:6}
 For sets $\mathcal P_1,\dots,\mathcal P_n$ of probabilities, we have
 $$
 \dsquare_{i=1}^n \sup_{\mathbb Q\in \mathcal P_i} \VaR^{\mathbb Q}_{\alpha_i} (X)
 = \inf\{x\in \R: \{X>x\}\in \mathcal C_v\},
 $$
 where $$
 \mathcal C_v=\left\{\bigcup_{i=1}^n A_i: \sup_{\mathbb Q\in \mathcal P_i} \mathbb Q(A_i) \le \alpha_i,~i\in[n] \right\}.
 $$
 \end{corollary}
 In particular, when $\mathcal P_1,\dots,\mathcal P_n$ are the same singleton $\{\p\}$ for an atomless $\p$, Corollary \ref{co:6} yields 
\eqref{eq:ELW18}, by noting that 
\begin{equation}
   \label{eq:ELW18-3}
 \mathcal C_v=\left\{\bigcup_{i=1}^n A_i:  \p(A_i) \le \alpha_i,~i\in[n] \right\}= \left\{A\in \mathcal F: \p(A)\le \sum_{i=1}^n\alpha_i\right\}.
\end{equation}
 Note that the non-atomicity assumption of $\p$ is needed for \eqref{eq:ELW18-3}, because to write $A$  with $\p(A)=\sum_{i=1}^n\alpha_i$
 as the union of $A_1,\dots,A_n$ with $\p(A_i)\le \alpha_i$ for each $i$,
 we need these probability constraints to be binding. Nevertheless, for the results in Theorem \ref{thm:inf} and Corollary 
\ref{co:6}, we do not need the non-atomicity assumption.

We can further notice that when $\mathcal P_1,\dots,\mathcal P_n$ are finite sets, the set $\mathcal C_v$ is closed from below, and an optimal allocation exists in the form of \eqref{eq:optimalsolution} in  Theorem \ref{thm:inf}.

    One may wonder whether \eqref{eq:ELW18} also holds when $\p$ is replaced by capacity, that is, 
 for $\alpha_i\in (0,1)$ and $\sum_{i=1}^n\alpha_i<1$, whether the equality
  \begin{equation}
    \label{eq:ELW18-2}  \dsquare_{i=1}^n \VaR_{\alpha_i}^{w}  (X) =\VaR_{\sum_{i=1}^n\alpha_i}^{w}(X)
    \end{equation}
    holds.
    This is generally not true, as discussed below. 
    A  formula similar to \eqref{eq:ELW18-2} holds true in the special case that $w$ is  a distorted probability.
\begin{corollary} \label{co:4} For $\alpha_1,\dots,\alpha_n\in (0,1)$  and $X\in\X$, let $g_i:[0,1]\to [0,1]$ be strictly increasing and continuous functions satisfying $g_i(0)=0$ and $g_i(1)=1$ for $i\in [n]$, and $\mathbb Q$ be an atomless probability. We have
\begin{align*}
    \dsquare_{i=1}^n\VaR_{\alpha_i}^{g_i\circ \mathbb Q}(X)=\VaR_{\sum_{i=1}^n g_i^{-1}(\alpha_i)}^\mathbb Q(X).
 \end{align*}
 \end{corollary}


   If $g_1=\dots=g_n=g$, by writing  $\alpha =\sum_{i=1}^n\alpha_i$, we have  
   $$\dsquare_{i=1}^n\VaR_{\alpha_i}^{g\circ \mathbb Q}(X)=\VaR_{\sum_{i=1}^n g^{-1}(\alpha_i)}^{\mathbb Q}(X) 
 \mbox{~~and~~}\VaR_{\alpha}^{g\circ \mathbb Q}(X)=\VaR_{g^{-1}(\alpha )}^{\mathbb Q}(X).   
   $$   
Generally, because $\sum_{i=1}^n g^{-1}(\alpha_i)= g^{-1}(\alpha)$ does not always hold, we have $\dsquare_{i=1}^n\VaR_{\alpha_i}^{g\circ \mathbb Q}(X)\neq \VaR_{\alpha}^{g\circ \mathbb Q}(X)$.
This implies that $\dsquare_{i=1}^n\VaR_{\alpha_i}^{w}(X)$ and $ \VaR_{\alpha}^{w}(X)$ are not equal for a capacity $w$, different from the case of probabilistic quantiles.

Before ending this section, we discuss the structure of the  composition $(A_1^*,\dots, A_n^*)\in \Pi_n$ in the optimal allocation \eqref{eq:optimalsolution} for a special case
of agents using quantiles with sup-probabilities; that is, each agent $i\in [n]$ uses the risk measure 
 $\sup_{\mathbb Q\in \mathcal P_i} \VaR^{\mathbb Q}_{\alpha_i} 
$ 
for some set $\mathcal P_i$ of probabilities   and $\alpha_i\in (0,1)$.
We consider the following two standard probability spaces: 
$(\Omega,\mathcal F,\p)$ is either an atomless space or an equiprobable space $([m],2^{[m]},\p)$ with $\p(\{k\}) =1/m$ for $k\in [m]$. We assume
that all probability measures in $\mathcal P_1,\dots,\mathcal P_n$ are  absolutely continuous with respect to $\mathbb P$.


\begin{proposition} \label{prop:9}  In the above setting, fix $X\in\mathcal X$. Suppose
that  ${\rm d}\mathbb Q/{\rm d}\mathbb \p $ is decreasing in $X$ if $\mathbb Q\in\mathcal P_1$ and increasing in $X$ if $\mathbb Q\in\mathcal P_i$, $i=2,\ldots,n$ and the optimal allocation  \eqref{eq:optimalsolution} exists. Then the optimal composition $(A_1^*,\dots, A_n^*)\in \Pi_n$ can be chosen such that ${A_1^*}$ is a tail event of $X$, meaning $X(\omega)\ge X(\omega')$ for $\omega\in A_1^*$ and $\omega'\not\in A_1^*$.
\end{proposition}

The monotone likelihood ratio condition  on $\d \mathbb Q/\d \p$ is important for this result. 
Intuitively, Proposition \ref{prop:9} means that 
if agent $1$ thinks that some certain tail event $A$ of $X$ is the least likely compared to other events under each $\mathbb Q$ in her belief set $\mathcal P_1$  (implied by the monotone likelihood ratio condition), and other agents think that the tail event $A$ is more likely compared to other events, then this tail event is allocated to agent $1$.

\section{Choquet Expected Shortfall}\label{sec:ES}

As one of the most important risk measures, the standard risk measure in banking regulation is the class of Expected Shortfall (ES),  closely related to probabilistic quantiles.   We refer to  \cite{MFE15} and \cite{FS16} for more discussions on ES as the standard regulatory risk measure.  
Since Choquet quantiles share many properties as probabilistic quantiles, it is natural to consider the corresponding ES as a new tool for risk management. We propose such a class of risk measures, study their properties, and discuss  their inf-convolution. 

\subsection{Definition and properties}
 
Define the \emph{Choquet Expected Shortfall (Choquet ES)} under a capacity $w$ by $$\ES_{\alpha}^{w}(X)=\frac{1}{\alpha}\int_0^{\alpha}\VaR^w_{t}(X)\d t,~~~X\in\X,~~~ \alpha\in (0,1].$$ If $w=\mathbb Q$ is a probability, $\ES_{\alpha}^{\mathbb Q}$ is the classic probabilistic ES.  
It is straightforward to verify that 
a Choquet ES satisfies comonotonic additivity (defined in Section \ref{sec:def}), monotonicity, translation invariance,  and positive homogeneity (defined in Section \ref{sec:property}), based on the corresponding properties of Choquet quantiles. 

It is well-known that a probabilistic ES is a coherent risk measure.   
We give an integral representation of Choquet ES, and show that it is a coherent risk measure under an additional condition. 
\begin{proposition}\label{prop:ESC}
For a capacity $w$ and $\alpha \in (0,1]$,
 $$\ES_\alpha^w=I_{w_\alpha}, \mbox{~where~} w_\alpha= (w /\alpha )\wedge 1.$$ Moreover,
 $\ES_{\alpha}^{w} $ is a coherent risk measure if and only if $w_{\alpha}$ is submodular.
 \end{proposition}

Using the fact that $w_{\alpha }$ is submodular implies that $w_{\beta}$ is  submodular  for $\beta \le \alpha$, we immediately arrive at the following conclusion.
\begin{corollary}\label{cor:ES} For a capacity $w$,  $\{\ES_{\alpha}^{w}\}_{\alpha\in (0,1]}$ is a decreasing family of coherent risk measures if and only if $w$ is submodular.
\end{corollary}
Note that if $\alpha=1$, $\ES_{1}^{w}$ coincides with  Choquet integral $I_w$.  
Moreover, if $v$ is a probability $\p$, then $\p_{\alpha} = (\p/\alpha )\wedge 1$ is submodular, which implies that  $\ES_{\alpha}^{\p}$ is a coherent risk measure, a well-known result in the literature.

For quantiles and ES under a probability $\p$, \cite{RU00, RU02} showed that they admit  an important  relation that a quantile and a corresponding ES can be respectively represented as the minimizers and the optimal value of a minimization problem whose objective function is a convex function in $\R$ and  linear in probability.
We next show that similar result holds for Choquet quantiles and Choquet $\ES$. 
\begin{theorem}\label{th:min} For a capacity $w$, $X\in\X$ and $\alpha\in (0,1)$, we have
\begin{align*}\left[\VaR_{\alpha}^w(X), \overline{\VaR}_{\alpha}^w(X)\right]&=\argmin_{x\in\R}\left\{x+\frac{1}{ \alpha}\int(X-x)_+\d w\right\},\\
\ES_{\alpha}^{w}(X)&=\min_{x\in\R}\left\{x+\frac{1}{ \alpha}\int(X-x)_+\d w \right\}.
\end{align*}
\end{theorem}


Next, we discuss the dual representation of Choquet ES.  Let $\mathcal P_0$ be the set of all probabilities on $(\Omega,\mathcal F)$.
For two functions $w,w'$ on $\mathcal F$, 
the inequality $w\le w'$ 
is understood point-wise. 
The classic dual representation of ES (see \citet[Theorem 4.52]{FS16}) is given by
$$
\ES_\alpha^\p (X) =\sup_{\mathbb Q\in\mathcal P_0}\left\{\E^{\mathbb Q}[X]: \mathbb Q\le \p /\alpha  \right\}.
$$
This representation also holds for   coherent Choquet ES.

\begin{proposition}\label{prop:ESS}
For a continuous capacity $w$,   $\alpha\in (0,1)$,
if $\ES_\alpha^w$ is coherent, then
$$
\ES_\alpha^w(X) =\sup_{ \mathbb Q\in\mathcal P_0 }\left\{\E^{\mathbb Q}[X]: \mathbb Q\le w/\alpha \right\},~~X\in \X.
$$
\end{proposition}

Denote by $\mathcal W$ the set of all capacities on $(\Omega,\mathcal F)$.
We note that 
$$
\ES_\alpha^w(X) =\sup_{ w'\in\mathcal W }\left\{I_{w'}(X) : w'\le w/\alpha \right\},~~X\in \X
$$ 
holds true regardless of any property of $w$ because $w_\alpha=(w/\alpha)\wedge 1$ is the maximum of the set $\{w'\in \mathcal W: w'\le w/\alpha \}$. Proposition \ref{prop:ESS} explains that the set $\mathcal W$ can be chosen as the set $\mathcal P_0$ of probabilities if $w_\alpha$ is submodular.  
  If the continuity of $w$ is not assumed, then the representation in Proposition \ref{prop:ESS} still holds by replacing $\mathcal P_0$ with  the collection of all finitely additive capacities on $(\Omega,\mathcal F)$.




\subsection{Inf-convolution}
The inf-convolution of several coherent Choquet integrals is obtained in the theory of risk measures; see \cite{BE05}. It has an explicit form summarized in the next proposition. 
We denote by $\mathcal S$ the set of all   submodular capacities.
\begin{proposition}\label{prop:coherentchoquet} For submodular capacities $w_1,\dots,w_n\in \mathcal S$, we have
$$
\dsquare_{i=1}^n I_{w_i} (X) =\sup_{ w\in\mathcal S }\left\{I_{w}(X) : w \le  \bigwedge_{i=1}^nw_i  \right\},~~X\in \X,
$$
where $\sup\varnothing=-\infty$.
\end{proposition}
 
Combining Propositions \ref{prop:ESC} and \ref{prop:coherentchoquet}, we immediately obtain an expression for the inf-convolution of Choquet $\ES$ with different capacities. 
\begin{corollary} 
\label{coro:ES-inf}
Let $\alpha_1,\dots,\alpha_n\in (0,1]$.
For submodular capacities $w_1,\dots, w_n \in \mathcal S$, we have
$$\dsquare_{i=1}^n \ES_{\alpha_i}^{w_i}(X)=\sup_{w\in \mathcal S}\left\{I_{w}(X):w\leq \bigwedge_{i=1}^n\frac{w_i}{\alpha_i}\right\},~~X\in \X,$$
where $\sup\varnothing=-\infty$.
In particular, 
for a submodular capacity $w$, we have
$$\dsquare_{i=1}^n \ES_{\alpha_i}^{w}=\ES_{\bigvee_{i=1}^n\alpha_i}^{w}.$$
\end{corollary}

We note that 
$
\sup_{w\in W} \ES_\alpha^w
$
is not a Choquet ES. As shown by \cite{WZ21}, the supremum of ES for different probabilities is not comonotonic additive; recall that a Choquet ES needs to be comonotonic additive by Proposition \ref{prop:ESC}. Generally, $\sup_{w\in W} \ES_\alpha^w$ is smaller than  the Choquet ES given by $\ES_\alpha^{\overline w}$, where $\overline w=\sup_{w\in W} w$. 
Therefore, the Choquet ES does not have the direct interpretation as ES under ambiguity, in contrast to Choquet quantiles, which include worst-case probabilistic quantiles under ambiguity.  

Proposition \ref{prop:coherentchoquet} and Corollary \ref{coro:ES-inf} require that the capacities $w_1,\dots,w_n$ to be submodular. 
When these capacities are not submodular, as in the case of general sup-probabilities, computing the inf-convolution of Choquet ES and the corresponding optimal allocation 
seems to be quite difficult.
In the next two sections, we focus on computation of risk sharing for Choquet quantiles with  sup-probabilities.

\section{Computation of risk sharing with sup-probabilities}\label{sec:numerical}

 We now numerically solve the problem of    risk sharing among agents using probabilistic quantiles with ambiguity on the underlying probability.  
 Suppose that agent $i\in [n]$ has  an ambiguity set $\mathcal P_i$, consisting of some probabilities on $(\Omega,\mathcal F)$, and 
 uses a probabilistic quantile with level $\alpha_i\in (0,1)$.
The risk measure for agent $i$ is   $\sup_{\mathbb Q\in \mathcal P_i} \VaR_{\alpha_i}^{\mathbb Q}$.  
 The risk sharing  problem  is to find optimal allocations for the problem
\begin{align}\label{eq:app1}
\dsquare_{i=1}^n \sup_{\mathbb Q\in \mathcal P_i} \VaR_{\alpha_i}^{\mathbb Q} (X)=\inf\left\{ \sum_{i=1}^n\sup_{\mathbb Q\in \mathcal P_i} \VaR_{\alpha_i}^{\mathbb Q}(X_i): (X_1,\cdots,X_n)\in \mathbb A_n(X)  \right\},
\end{align}
where $\mathbb{A}_n(X)$ is the set of all   allocations of $X$ given by \eqref{eq:intro1}.
As we explained in Section \ref{sec:inf-convolution}, solving \eqref{eq:app1} is equivalent to finding all Pareto-optimal allocations.

\subsection{Atomless setting}
We consider the setting in which $(\Omega,\mathcal F) = ([0,1],\mathcal B([0,1]))$ and $\p$ is the Lebesgue measure.  
For computational feasibility, suppose that    $\mathcal P_i$ is a finite set of   probability measures on $([0,1],\mathcal B([0,1]))$ absolutely continuous with respect $\p$ such that 
$\d \mathbb Q/\d \p$ takes finitely many values for $\mathbb Q \in \mathcal P_i$.
Denote by $N_i$ the cardinality of $\mathcal P_i$.
Moreover, the total risk $X$ also takes finitely many values.
Under this setting, we can find  an $m$-partition $\{B_j\}_{j\in[m]}\subseteq \mathcal B([0,1])$  of $\Omega$, generated jointly by all elements of $\mathcal P_1,\dots,\mathcal P_n$ and $X$, 
such that we can write
$$X= \sum_{j=1}^m x_j \id_{B_j}$$
for some constants $x_1,\dots,x_m$ and 
$$
\mathcal P_i =\left\{ \mathbb Q_{k}^{(i)}: \frac{{\rm d} \mathbb Q_{k}^{(i)}}{{\rm d} \mathbb \p}=\sum_{j=1}^m p_{kj}^{(i)} \id_{B_j},~ k\in [N_i]  \right\},~~i\in [n],
$$
where $p_{kj}^{(i)}\ge 0$ satisfies $\sum_{j=1}^m p_{kj}^{(i)} \p(B_j)=1$ for $k\in [N_i]$ and $i\in [n]$.  
By  Corollary \ref{co:6},  problem \eqref{eq:app1} is equivalent to
 $$
 \dsquare_{i=1}^n \sup_{\mathbb Q\in \mathcal P_i} \VaR^{\mathbb Q}_{\alpha_i} (X)
 = \inf\{x\in \R: \{X> x\}\in \mathcal C_v\},
 $$
 where $$
 \mathcal C_v=\left\{\bigcup_{i=1}^n A_i: \sup_{\mathbb Q\in \mathcal P_i} \mathbb Q(A_i) \le  \alpha_i,~i\in[n] \right\}.
 $$ Therefore,  solving problem \eqref{eq:app1} is equivalent to finding the smallest $x\in\R$ such that 
\begin{align}\label{eq:app2}
\{X > x\} =\bigcup_{i=1}^n A_i~~{\rm with}~~\sup_{\mathbb Q\in \mathcal P_i} \mathbb Q(A_i)\le  \alpha_i,~~i\in [n].\end{align}
By straightforward manipulation, the above problem is equivalent to
\begin{align*}
    \min  ~~~& x\\
{\rm subject~to}~~ & ( x -x_j) \id_{\{\sum_{i=1}^n s_{j}^{i} <1\}}\ge 0,~~j\in [m],\\
  & \sum_{j=1}^m p_{kj}^{(i)} \p(B_j)s_{j}^{i} \le  \alpha_i,\,k\in [N_i],\,i\in [n],\\
  & s_{j}^{i} \in [0,1], \,j\in [m], \,  i\in [n].
\end{align*}
The optimal value satisfies  $x^*= \max_{j\in [m]}x_j\id_{\{\sum_{i=1}^n(s_j^{i})^*<1\}}$,
where 
$(x ^*, (s_j^{i})^*_{i\in[n], j\in [m]})$ is a corresponding minimizer.
One optimal allocation is given by
$$
X_i^* = (X-x^*)\id_{A_i^*} +\frac{x^*}n,~~i\in [n],
$$
where  $A_i^* = \bigcup_{j=1}^m \left(B_j\cap {\{U\in [\sum_{k=0}^{i-1}(s_{j}^k)^*,\sum_{k=0}^{i}(s_{j}^k)^*]\} }\right)$, $i\in [n-1],$ and $A_n^*=\Omega\setminus (\bigcup_{i=1}^{n-1}A_i^*)$ with  $(s_{j}^0)^*=0$ for $j\in [m]$ and  $U\sim {\rm U}[0,1]$ being independent of $(B_j)_{j\in [m]}$.

Although we provided a specific form of the allocation,  on the event $\{X<x^*\}$  the value of $X_i^*$ is not important, which  can be arbitrary as long as it is not larger than $x^*/n$ due to the definition of Choquet quantiles.

We provide a few numerical examples. 
In each example, we provide the optimal allocation visually. 
We also compare the optimal total risk (post-sharing), represented by the inf-convolution in \eqref{eq:app1}, and the total   risk before sharing, represented by $\sum_{i=1}^n \sup_{\mathbb Q\in \mathcal P_i} \VaR_{\alpha_i}^{\mathbb Q} (X/n)$, assuming that all agents have an initial position $X/n$.
The post-sharing risk, as an optimum, is smaller than or equal to the latter value, and we are interested in how much risk reduction is achieved by the optimal risk sharing.
The reduction is computed as one minus the ratio of the  (optimal) total  post-sharing risk to the initial risk.


\begin{table}[t]
    \begin{center}
    \caption{Risk comparison before and after risk sharing in the atomless setting}
    \label{tab:1} \renewcommand{\arraystretch}{1.5} 
    \begin{tabular}{c|c|cccc} 
        $n$  & $(\alpha_1,\dots,\alpha_n)$ & model  & Initial risk & Post-sharing risk  &   Reduction  \\ \hline
        \multirow{6}{*}{{2}} & \multirow{2}{*}{(0.05, 0.1)} & DU  & 0.95 & 0.85 & 11\%  \\ 
         &  & DP & 10.00 &  5.00 & 50\% \\ \cline{2-6}
         & \multirow{2}{*}{(0.1, 0.1)} & DU  & 0.95 & 0.80 & 16\%  \\ 
         &  & DP & 10.00 &  4.00 & 60\% \\ \cline{2-6}
          & \multirow{2}{*}{(0.3, 0.1)} & DU  &  0.85& 0.70 & 18\%  \\ 
         &  &DP &6.67   & 2.86 & 57\%  \\ \hline
        \multirow{4}{*}{{5}} & \multirow{2}{*}{(0.1, 0.1, 0.05, 0.05, 0.05)} & DU   &  0.98 & 0.40  &  59\% \\ 
         &  & DP & $16.00$ & 1.54 & 90\%  \\  \cline{2-6} 
         & \multirow{2}{*}{(0.05, 0.05, 0.025, 0.025, 0.025)} &  DU   &  0.98 & 0.75  &  24\% \\
         &  &DP & $16.00$ & 3.33 & 79\% \\ \hline
    \end{tabular}\\
    \label{tab:example1}
    \end{center}
    \small {\it Notes.} DU and DP represent   discrete uniform and discrete power models, respectively.
    Let 
    $Y$  have a uniform distribution on the set $\{j/20: j=1,\dots,20\}$ under $\p$.
    In the DU model,   $X$ is identically distributed to $Y$. 
    In the DP model, 
       $X$ is distributed as $1/Y$.       
       The initial risk and the post-sharing risk are defined as $\sum_{i=1}^n \sup_{\mathbb Q\in \mathcal P_i} \VaR^{\mathbb Q}_{\alpha_i} (X/n)$ and $\dsquare_{i=1}^n \sup_{\mathbb Q\in \mathcal P_i} \VaR^{\mathbb Q}_{\alpha_i} (X)$, respectively.
       The probability measures in $\mathcal P_1,\dots,\mathcal P_n$ are randomly generated. The reduction is computed  as one minus the ratio of the post-sharing risk to the initial risk.
\end{table}


    \begin{figure}[t]
            \centering
              \begin{subfigure}[b]{0.99\textwidth}
            \centering
\includegraphics[width=12.5cm]{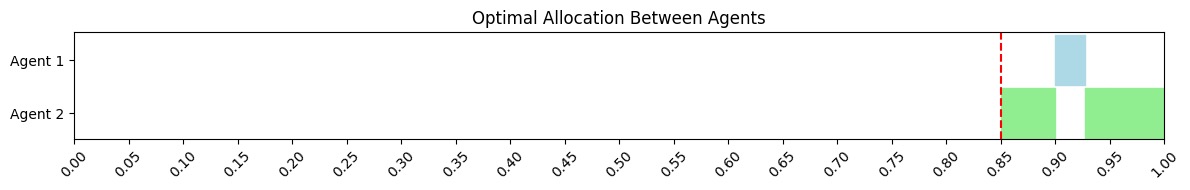} 
    \caption{$(\alpha_1,\alpha_2)=(0.05,\,0.1)$} 
    \end{subfigure}\\~\\
    \begin{subfigure}[b]{0.99\textwidth}
            \centering
\includegraphics[width=12.5cm]{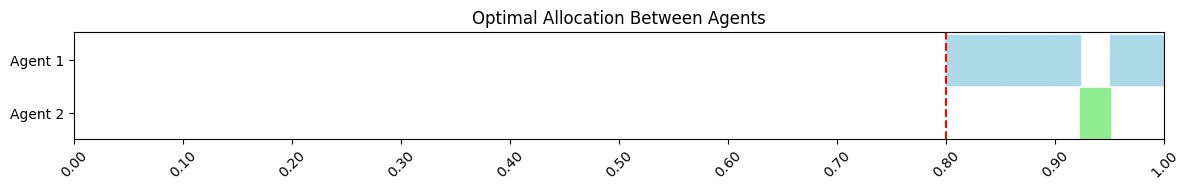} 
    \caption{$(\alpha_1,\alpha_2)=(0.1,\,0.1)$} 
    \end{subfigure}\\~\\
     \begin{subfigure}[b]{0.99\textwidth}    \centering
    \includegraphics[width=12.5cm]{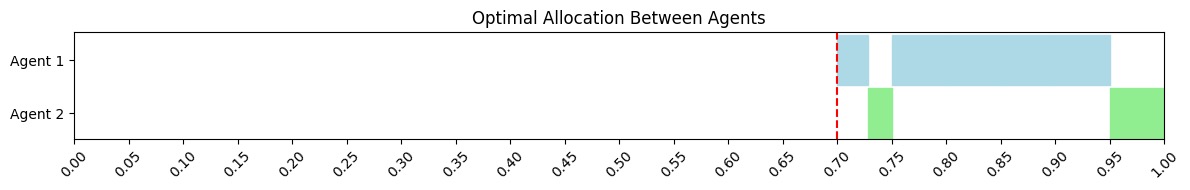}     
    \caption{$(\alpha_1,\alpha_2)=(0.3,\,0.1)$} 
    \end{subfigure} 
    \caption{Optimal allocation between two agents in the atomless setting in Example \ref{ex:3}
}\label{fig:11-atomless} 
\end{figure}

 Let 
    $Y$  have a uniform distribution on the set $\{j/m: j=1,\dots,m\}$ with $m=20$ under $\p$.
    In the discrete uniform model, the random variable $X$ is identically distributed to $Y$ under $\p$. 
    In the discrete power model,
       $X$ is distributed as $1/Y$.
       To better present the plots, we assume that $X(\omega)$ is increasing in $\omega$ without loss of generality. 
       By Proposition \ref{prop:invariance}, the optimal allocations for the above two models are very similar, because   the risks $X$ in both models can be written as increasing transforms of each other. Especially, they share the same composition $(A_1^*,\dots,A_n^*)$, which we present in the figures.

\begin{example} \label{ex:3} Let $n=2$ and $m=20$, $B_j=[(j-1)/m,j/m)$, $j\in [m]$, and   $N_i=3$ for $i=1,2$  and $X=\sum_{j=1}^m j/m\id_{B_j}$,  where $p_{kj}^{(i)}$ are generated randomly and uniformly such that $\sum_{j=1}^m p_{kj}^{(i)}/m=1$. 
Rigorously speaking, $\{B_j\}_{j\in [m]}$ is a partition of $[0,1)$ instead of $[0,1]$, but this creates no difference under the Lebesgue measure. 
The optimal value (red line) and allocation (rectangles) are presented in Figure  \ref{fig:11-atomless} for two different choices of $(\alpha_1,\alpha_2)$.
Note that although the  set $B_j$ in the partition has the length $1/m$, the optimal allocation has events that are of smaller length, meaning that each event $B_j$ may still be split among the two agents.  
The case of $n=5$ agents is presented in Figure \ref{fig:22} in Section \ref{app:numerical} for two different choices of $(\alpha_1,\dots,\alpha_5)$.
Table \ref{tab:1} illustrates the risk measure values before and after risk sharing under two models. 
We can make two intuitive observations from 
 Table \ref{tab:1}. 
 First, the risk reduction is larger in percentage for the discrete power model than the discrete uniform model. The discrete power model has a more skewed distribution, thus more ``risky'' in the usual sense (it can be seen as a discrete version of the Pareto distribution), and risk sharing is more efficient for such losses. 
 Second, the risk reduction is larger in percentage   when the number of agents is larger. This confirms the intuition that   the risk sharing mechanism becomes more efficient when more agents join the pool. 
\end{example}
  
\begin{example}
\label{ex:3-2}
Under the assumption of Example \ref{ex:3}, but now let $n=3$ and  $\{p_{kj}^{(1)}: j\in [m]\}$ be  sorted  in the descending order 
for each $k\in [N_1]$
and $\{p_{kj}^{(i)}: j\in [m]\}$  be  sorted  in the ascending order  for each $k\in [N_i]$ and $i\ge 2$.  The optimal value (red line) and allocation (rectangles) are presented in Figure  \ref{fig:32-atomless} for two different choices of $(\alpha_1,\alpha_2,\alpha_3)$. We can see that in the  optimal allocation,  $A_1^*$ takes the form of $[\omega_1,1]$ for some $\omega_1\in [0,1]$, as theoretically justified by  Proposition \ref{prop:9}. 
\end{example}

   \begin{figure}[t]
    \centering
      \begin{subfigure}[b]{0.99\textwidth}    \centering
\includegraphics[width=12.5cm]{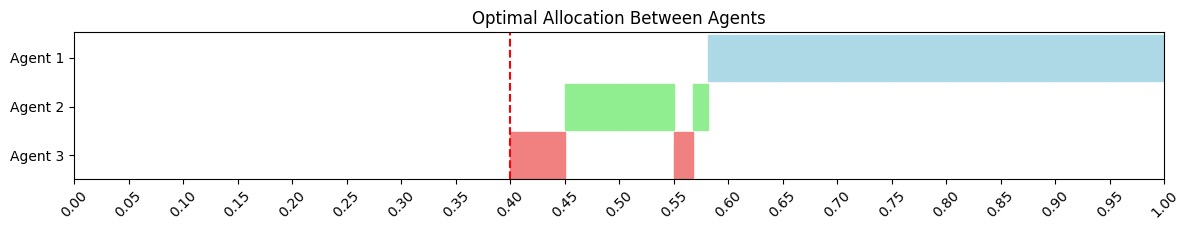} 
    \caption{$(\alpha_1,\alpha_2,\alpha_3)=(0.1, 0.1, 0.05)$} 
    \end{subfigure}  \\~\\
      \begin{subfigure}[b]{0.99\textwidth}
    \centering
    \includegraphics[width=12.5cm]{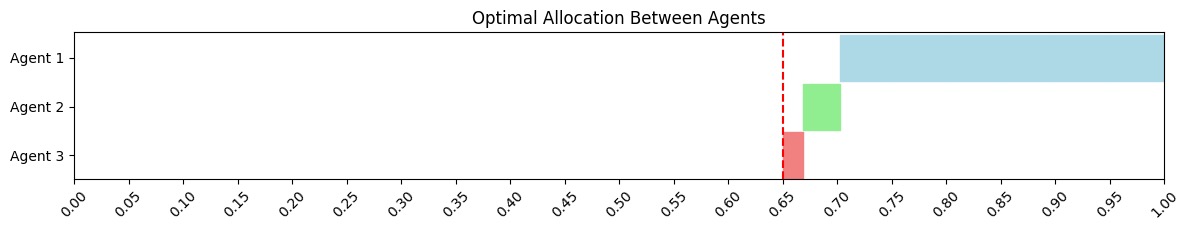}
          \caption{$(\alpha_1,\alpha_2,\alpha_3)=(0.05, 0.05, 0.025)$}\end{subfigure}
    \caption{Optimal allocation among three agents in the atomless setting in Example \ref{ex:3-2}}
    \label{fig:32-atomless}
\end{figure} 


\subsection{Discrete setting}
\label{sec:62}

Next, we consider the discrete setting, where
$\Omega=\{\omega_1,\dots,\omega_m\}$ and $\mathcal F$ is the set of all subsets of $\Omega$.
Suppose that $\mathcal P_i$ is a finite set of discrete probability measures on $(\Omega,\mathcal F)$, that is,  
$$
\mathcal P_i =\left\{\mathbb Q_k^{(i)}=\sum_{j=1}^m p_{kj}^{(i)} \delta_{\omega_j}: k\in [N_i]  \right\},~~i\in [n],
$$
where $p_{kj}^{(i)}\ge 0$ satisfies $\sum_{j=1}^m p_{kj}^{(i)}=1$ for $k\in [N_i]$ and $i\in [n]$, and $\omega_j\in\Omega$ for $j\in [m]$, and $\delta_\omega$ is the Dirac measure at $\omega$. 
For a given random variable $X\in\mathcal X$, 
 we assume 
$$
X(\omega_{j_1})< X(\omega_{j_2})< \cdots< X(\omega_{j_m}).
$$ 
Define  
$$ q_{k\ell }^{(i)} = p_{kj_\ell}^{(i)},~~\ell\in [m].$$ 
Based on \eqref{eq:app2}, we have that problem \eqref{eq:app1} is equivalent to the following optimization problem:
\begin{align*}
    \min_{\ell, \, {\bf s}}~&~x\\
    {\rm s.t.}~
    &~\ell\in [m],~x\ge X(\omega_{j_{\ell-1}}),\\
    &~q_{k\ell }^{(i)} s_{i\ell}+\cdots+ q_{km}^{(i)}  s_{im}\le  \alpha_i,~~k\in [N_i],~i\in [n],\\
    & ~s_{ij}\in \{0,1\},\,i\in[n],\,  ~s_{1j}+\dots+s_{nj}=1,\,j\in [m],
\end{align*}
where $X(\omega_{j_0})=-\infty$.
If the minimizer to the above problem is $(\ell^*, (s_{i\ell}^*)_{i\in[n], j\in [m]})$, then the optimal value is $x^*=X(w_{j_{\ell^*-1}})$ and one optimal allocation is given by
$$
X_i^* = (X-x^*)\id_{A_i^*} +\frac{x^*}n,~~i\in [n],
$$
where $(A_1^*,\dots, A_n^*)\in \Pi_n$  and  $ \{\omega_\ell: s_{i\ell}^*=1,~\ell=\ell^*,\ldots,m\}\subseteq A_i^* $ for $i\in [n].$

The main technical difference between the discrete setting
and the atomless setting is that each atom now can only be allocated to one agent, and there is no split between agents on a single atom.


\begin{table}[t]
    \begin{center}
    \caption{Risk comparison before and after risk sharing in the atomless setting; for details, see the notes of Table \ref{tab:example1}}  \renewcommand{\arraystretch}{1.5} 
    \begin{tabular}{c|c|cccc} 
        $n$  & $(\alpha_1,\dots,\alpha_n)$ & model  & Initial risk & Post-sharing risk  &   Reduction  \\ \hline
        \multirow{4}{*}{{2}} & \multirow{2}{*}{(0.1, 0.1)} & DU & 1.00 & 0.85 & 15\%  \\ 
         &  & DP & 20.00 &  5.00 & 75\% \\  \cline{2-6} 
          & \multirow{2}{*}{(0.3, 0.1)} & DU &  0.90 & 0.75 & 17\%  \\ 
         &  & DP & 12.00  & 3.33 & 72\%\\ \hline
          \multirow{2}{*}{{5}} & \multirow{2}{*}{(0.1, 0.1, 0.05, 0.05, 0.05)} & DU &  0.79 & 0.70  &  24\% \\ 
         &  & DP & 6.71 & 2.86 & 57\% \\ \hline
    \end{tabular}\\
    \label{tab:example2}
    \end{center} 
\end{table}

\begin{example}
\label{ex:4-1}
Let $n=2$ and $m=20$ and $\omega_j=j/m$ for $j\in [m]$, and we set $\mathcal P_i=\{\mathbb Q_{i1},\mathbb Q_{i2}, \mathbb Q_{i3}\}$, $i=1,2$ with $\mathbb Q_{ik}= \sum_{j=1}^m p_{kj}^{(i)} \delta_{j/m}$ for $i=1,2$ and $k=1,2,3$,  where $p_{kj}^{(i)}$ are generated randomly such that $\sum_{j=1}^m p_{kj}^{(i)}=1$.  For  $X$ with $X(\omega_j)=\omega_j$, $j\in [m]$, the optimal value (red line)  and allocation (rectangles) are presented in Figure \ref{fig:12} for two different choices of $(\alpha_1,\alpha_2)$.
The case $n=5$ with 
$(\alpha_1,\dots,\alpha_5) = (0.1, 0.1, 0.05, 0.05, 0.05)$ is depicted in Figure \ref{fig:2} in Section \ref{app:numerical}.
\end{example}
    \begin{figure}[t]
    \centering      \begin{subfigure}[b]{0.99\textwidth}    \centering   
    \includegraphics[width=12.5cm]{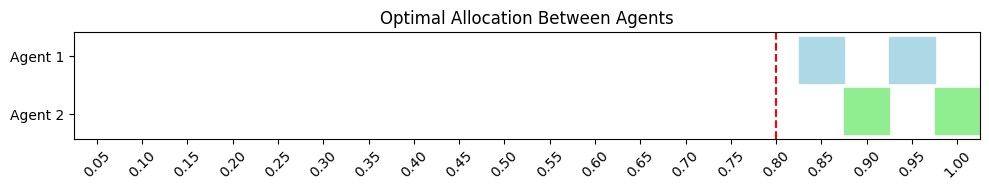} 
    \caption{$(\alpha_1,\alpha_2)=(0.1,0.1)$} 
     \end{subfigure}  \\~\\
      \begin{subfigure}[b]{0.99\textwidth}
    \centering
    \includegraphics[width=12.5cm]{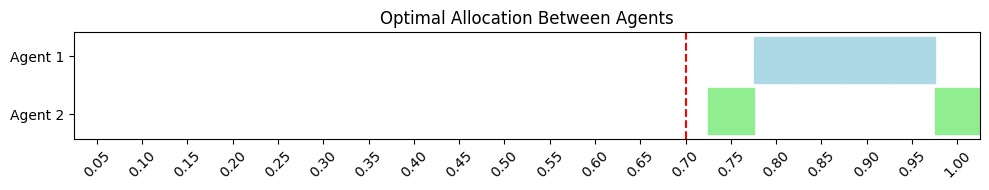}
   \caption{$(\alpha_1,\alpha_2)=(0.3,0.1)$}
         \end{subfigure} 
         \caption{Optimal allocation between two agents in the discrete setting in Example \ref{ex:4-1}; each block represents a discrete value (for instance, in the top panel, the green blocks represent two discrete points at $0.9$ and $1.0$ allocated to agent $2$)}
    \label{fig:12}
\end{figure}

\section{A financial data implementation}\label{sec:realdata}

In this section, we apply the our result on risk sharing with quantiles under ambiguity to the loss data of S\&P 500 index.  We collect   S\&P 500 index daily return data from January 1, 2014  to December 31, 2023. The daily log losses, given by $-\log r_t$, where $r_t$ is the daily return at day $t$, are used to capture the relative changes between consecutive trading days. Using the log loss to compute risk measures is a standard practice in risk management; see \cite{MFE15}.  We then divide  the data into different periods of five-year, two-year, and one-year intervals to estimate the distribution of the future (next day) log loss denoted by $L$. We compute the empirical distributions within each period and also fit two normal and t distributions. 
The probability measures  in the ambiguity sets used by the four agents    are represented by the estimated distributions of $L$, formulated in the following way. 
\begin{enumerate}
    \item The ambiguity set $\mathcal P_1$ of  agent 1 contains two probabilities, each computed based on five consecutive years of data, shown in Figure \ref{fig:RD-01}.  
    \item The ambiguity set $\mathcal P_2$ of  agent 2 contains five probabilities, each computed based on two consecutive years of data, shown in Figure \ref{fig:RD-02}.
        \item The ambiguity set $\mathcal P_3$ of  agent 3 contains $10$ probabilities, each computed based on one year  of data, shown in Figure \ref{fig:RD-SAA-1year} in Appendix \ref{app:numerical}.
            \item The ambiguity set $\mathcal P_4$ of  agent 4 contains 20 probabilities, which are fitted to normal and t distributions (by maximum likelihood estimation), each  based on one  year of data, shown partially in Figure \ref{fig:RD-normalt}. The full details are in
            Figure \ref{fig:RD-normalt-all} in Appendix \ref{app:numerical}.
\end{enumerate}

The approach of computing the loss distributions based on yearly data and then take the worst-case risk evaluation is specified in the Basel Accords (\cite{BASEL19}); see \citet[Section 1.1]{WZ21} for a detailed explanation.
The four agents in our setting follow such an approach, although using different windows of the data. 
The different window sizes reflect the trade-off between the stability of the financial loss data across different years 
and the error of statistical estimation.

 \begin{figure}[t]
    \centering
    \includegraphics[width=6cm]{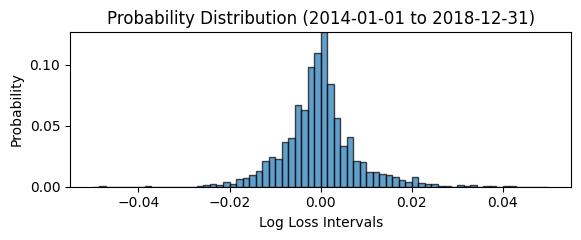} 
    \includegraphics[width=6cm]{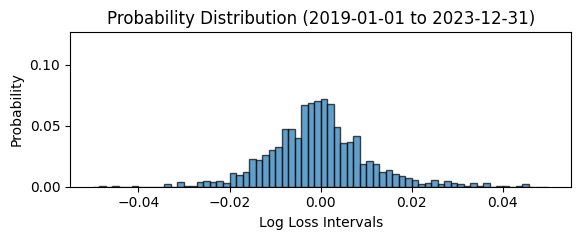} 
    \caption{Two probabilities  for the log loss of agent $1$}\label{fig:RD-01}
\end{figure}  

\begin{figure}[t]
    \centering
    \includegraphics[width=6cm]{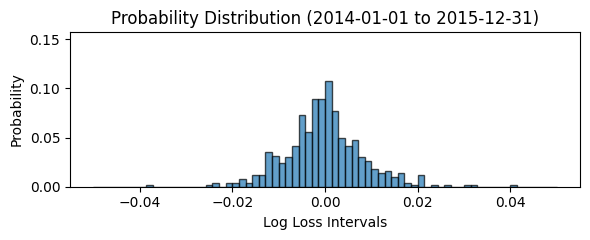} 
    \includegraphics[width=6cm]{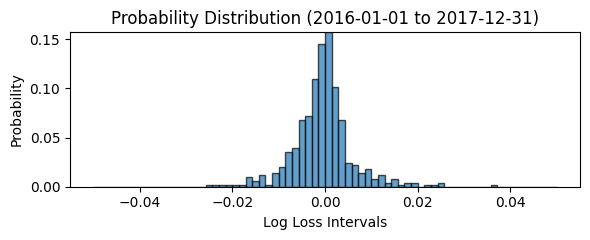} 
     \includegraphics[width=6cm]{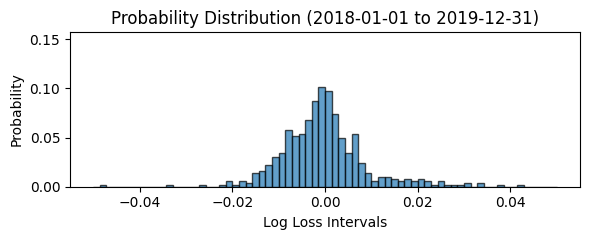} 
    \includegraphics[width=6cm]{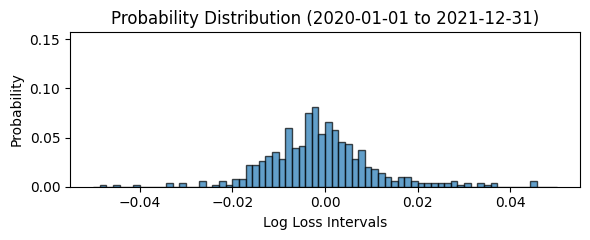} 
    \includegraphics[width=6cm]{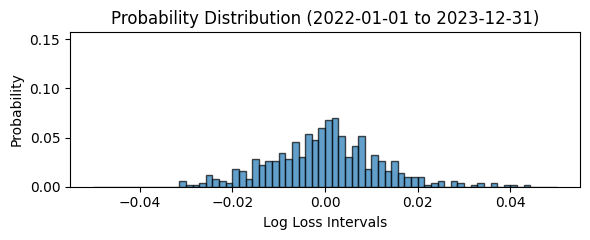} 
    \caption{Five  probabilities for the log loss of agent $2$}\label{fig:RD-02}
\end{figure}

 \begin{figure}[t] 
    \centering
    \includegraphics[width=12cm]{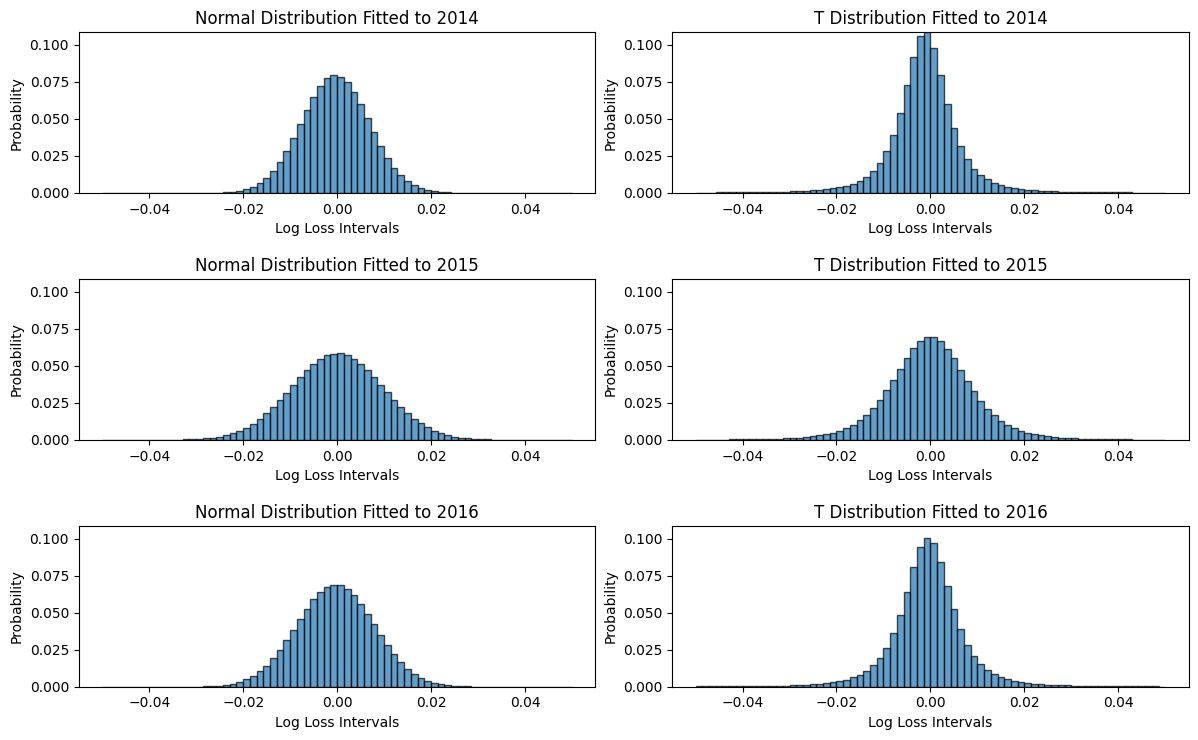} 
    \caption{Probabilities  for the log loss of agent $4$ fitted to normal and t distributions in the first 3 years}\label{fig:RD-normalt}
\end{figure}  

The total loss $X$ to share by the four agents is  the future loss from investing one dollar in the S\&P 500 index, that is, $X=1-\exp(-L)$, which is close to $L$ because the values of $L$ are small.
The  initial positions are assumed to be $X/4$ for each agent.  
To our numerical   method to compute the optimal risk sharing, we approximate the distributions with 700 equispaced discrete grids in $[-0.05,0.05]$ and use the method in Section \ref{sec:62}.  
The optimal  allocations are presented in Figure 
\ref{fig:RD-3}, corresponding to different values of $(\alpha_1,\alpha_2,\alpha_3,\alpha_4)$. Table \ref{tab:example} reports the initial risk
$\sum_{i=1}^4 \sup_{\mathbb Q\in \mathcal P_i} \VaR_{\alpha_i}^{\mathbb Q} (X/4)$
and the post-sharing risk
$\dsquare_{i=1}^4 \sup_{\mathbb Q\in \mathcal P_i} \VaR^{\mathbb Q}_{\alpha_i} (X)$ and the ratio of reduction, similarly to  Section \ref{sec:numerical}.
From the numerical results, we can see that by entering the risk sharing pool, agents are able to reduce their risk by 45\% in the case where $\alpha_1=\dots=\alpha_4=0.01$, which is a common choice of the VaR parameter for market risk.


   \begin{figure}[t]
     \centering      \begin{subfigure}[b]{0.99\textwidth}    \centering   
    \includegraphics[width=12.5cm]{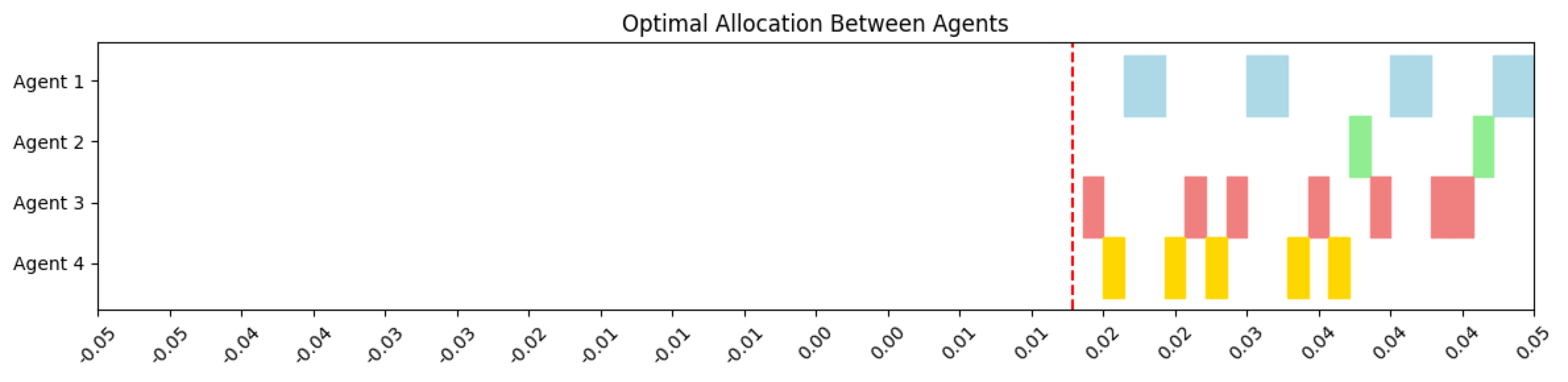} 
    \caption{$(\alpha_1,\ldots,\alpha_4)=(0.01, 0.01, 0.01, 0.01)$.} 
     \end{subfigure}  \\~\\ 
    \centering      \begin{subfigure}[b]{0.99\textwidth}    \centering   
    \includegraphics[width=12.5cm]{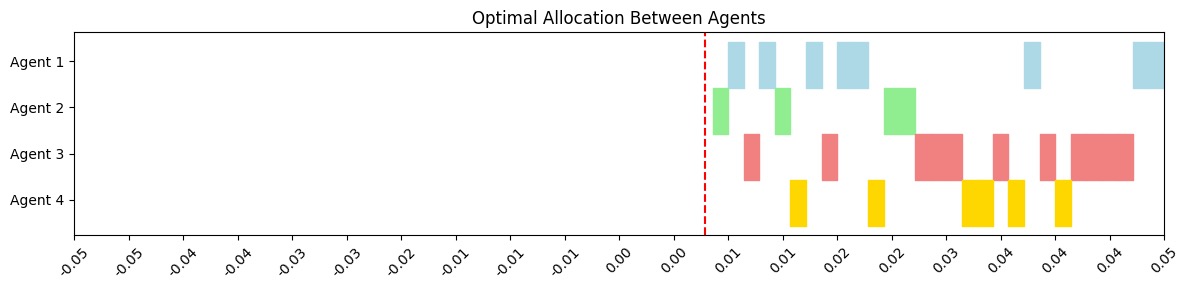} 
    \caption{$(\alpha_1,\ldots,\alpha_4)=(0.05, 0.05, 0.05, 0.05)$.} 
     \end{subfigure}  \\~\\
      \begin{subfigure}[b]{0.99\textwidth}
    \centering
    \includegraphics[width=12.5cm]{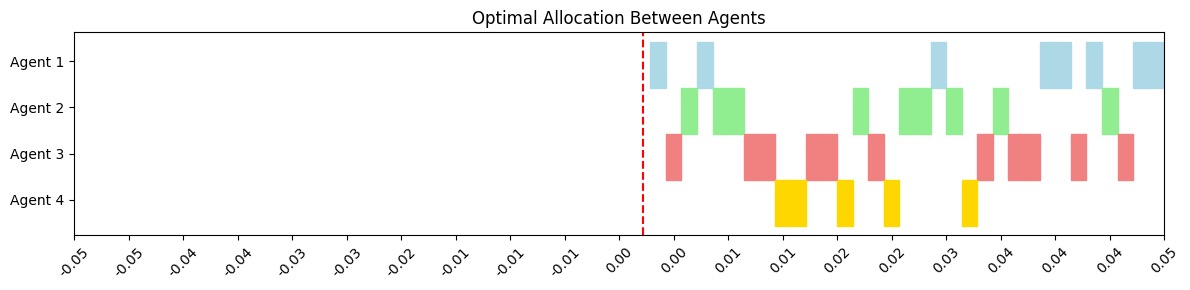}
   \caption{$(\alpha_1,\ldots,\alpha_4)=(0.1, 0.1, 0.1, 0.1)$.}
         \end{subfigure}  
         \caption{Optimal allocation of loss from the S\&P 500 index among four agents, where the space $\Omega$ is the set of values taken by $L$}
    \label{fig:RD-3}
\end{figure}  

\begin{table}[t]
    \begin{center}
    \caption{Risk Comparison before and after risk sharing for S\&P500 index}
    \renewcommand{\arraystretch}{1.5} 
    \begin{tabular}{c|ccc}
       $(\alpha_1,\dots,\alpha_4)$   & Initial risk & Post-sharing risk  &   Reduction  \\ \hline  
          (0.01, 0.01, 0.01, 0.01)   &0.0334 & 0.0184 & 45\%\\ \hline
       (0.05, 0.05, 0.05, 0.05)  & 0.0205 & 0.0085 &  59\%  \\ \hline
         (0.1, 0.1, 0.1, 0.1)     & 0.0152 &  0.0029 & 81\% \\ \hline  
    \end{tabular}\\
    \label{tab:example}
    \end{center}
\end{table}


\section{Conclusion}

\label{sec:conclusion}

Motivated by risk sharing problems for quantile agents under ambiguity, we formally introduced Choquet quantiles, which are interpreted as quantiles under ambiguity. We use only one axiom, ordinality, to characterize Choquet quantiles and find some nice properties enjoyed by Choquet quantiles such as closedness under the sup and inf operations and closedness under the inf-convolution. Those properties help us to fully solve the risk sharing problem for quantiles under heterogeneous beliefs and ambiguity. We further propose and study Choquet  ES, which is a version of ES under ambiguity.  Those two risk measures are useful tools to capture ambiguity and model uncertainty, which would be applicable to other related problems when ambiguity is involved.

\newpage
\appendix

\setcounter{table}{0}
\setcounter{figure}{0}
\setcounter{equation}{0}
\renewcommand{\thetable}{A.\arabic{table}}
\renewcommand{\thefigure}{A.\arabic{figure}}
\renewcommand{\theequation}{A.\arabic{equation}}

\setcounter{theorem}{0}
\setcounter{proposition}{0}
\renewcommand{\thetheorem}{A.\arabic{theorem}}
\renewcommand{\theproposition}{A.\arabic{proposition}}
\setcounter{lemma}{0}
\renewcommand{\thelemma}{A.\arabic{lemma}}

\setcounter{corollary}{0}
\renewcommand{\thecorollary}{A.\arabic{corollary}}

\setcounter{remark}{0}
\renewcommand{\theremark}{A.\arabic{remark}}
\setcounter{definition}{0}
\renewcommand{\thedefinition}{A.\arabic{definition}}

\section{Appendix: Technical proofs}
\label{app:A}
\noindent{\it Proof of Proposition \ref{prop:basic-def}.}
First note that, for $X\in \X$ and a capacity $v$ taking values in $\{0,1\}$, it holds that \begin{align} I_v(X)=\sup\{x\in \R: v(X\ge x)=1\}= \inf\{x\in \R: v(X\ge x)=0\} .\label{eq:CQ}\end{align}
This can be checked by the definition of Choquet quantile.
To show that \eqref{eq:rep1} and \eqref{eq:rep2} hold, it suffices to take $\alpha =1/2$ and $w=v$ and use \eqref{eq:CQ}.
To show that \eqref{eq:rep1} is a Choquet quantile, let $v:\mathcal F\to \{0,1\}$ be defined as  $v(A) =\id_{\{w(A)>1-\alpha \}}$ for $A\in\mathcal F$.
We obtain from \eqref{eq:rep1} that
$$
Q^w_\alpha (X) =  \inf \{ x\in \R : w(X\ge x) \le 1-\alpha\}= \inf \{ x\in \R : v(X\ge x) =0\},$$
which is a Choquet quantile by \eqref{eq:CQ}.
The proof of that \eqref{eq:rep2} is a Choquet quantile is similar. 
\qed\medskip


\noindent{\it Proof of Theorem \ref{th:main}.}
We first check the implication (ii)$\Rightarrow$(i).
For any increasing and continuous function $f$, $x\in \R$ and $\epsilon>0$,
it holds that $  \{f(X)\ge f(x)+\epsilon \} \subseteq \{X\ge x\}\subseteq \{f(X)\ge f(x)\}  $.
Using \eqref{eq:CQ} and the continuity of $f$, we have
\begin{align*}
f(\cR(X))  & =   f(\inf\{x : v( X \ge x)=0\})\\
& =   \inf\{f(x) : v( X \ge x)=0\}\\
&\le   \inf\{f(x) : v( f(X)  \ge f(x))=0\} \\
& =\inf\{x : v( f(X)  \ge x)=0\} = \cR(f(X)).
\end{align*} 
Similarly,  $
f(\cR(X))
\ge   \inf\{f(x) \in \R: v( f(X)  \ge f(x) +\epsilon )=0\}   = \cR(f(X)) - \epsilon. $
Since $\epsilon$ is arbitrary, we have $f(\cR(X))= \cR(f(X))$, and thus ordinality holds.
 Below we show (i)$\Rightarrow$(ii). We first verify three facts.
\begin{enumerate}[(a)]
\item   $\cR(X)$ takes values only in the closure of $\{X(\omega): \omega \in \Omega\}$, denoted by $C_X$. If $\cR(X) \not \in C_X$, then there exists an interval $(s,t)$ such that $\cR(X)\in (s,t)$ and $C_X\cap (s,t)=\varnothing$. It follows that $f(X)=g(X)$ for all increasing and continuous $f$ and $g$ which agree outside $[s,t]$ but do not agree inside $(s,t)$, while $f(\cR(X))\ne g(\cR(X))$ leading to a contradiction to ordinality.
\item $\cR$ is comonotonic additive.
For any two random variables $X$ and $Y$ that are comonotonic, by Denneberg's Lemma (Proposition 4.5 of \cite{D94}), there exist continuous and increasing functions $f$ and $g$
such that $X=f(Z)$ and $Y=g(Z)$, where $Z=X+Y$.
Note that $(f+g)(x)=x$ for $x$ in the range of $Z$.
Therefore, using ordinality of $\cR$ and (a),
\begin{align*} \cR(Z) =  (f+g) (\cR(Z))&= f (\cR(Z)) + g(\cR(Z))   = \cR(f(Z)) +\cR(g(Z))=\cR(X)+ \cR(Y).
\end{align*}
\item  $\cR$ is positively homogeneous; that is, $\cR(\lambda X)=\lambda \cR(X)$ for all $X\in \X$ and $\lambda >0$. This follows directly from ordinality of $\cR$.
\end{enumerate}

Note that  (a) yields $\cR(c)=c$   for any $c\in \R$.
Let $\X_0$ be the set of all simple functions in $\X_{\mathcal A}$, i.e., those which take finitely many values.
Using Proposition 1 of \cite{S86}, (b) and (c) imply that $\cR$ can be represented as a Choquet integral $\cR(X)=\int X\d v$ on $\X_0$
for some     function $v:\mathcal F\to \R$ defined by $v(A)=\cR(\id_A)$ for $A\in \mathcal F$.
Note that $v(\varnothing)=0$ and $v(\Omega)=1$ follow  from $\cR(0)=0$ and $\cR(1)=1$.
Using (a), we know   $v(A)\in \{0,1\}$.
Next we show that $v$ is increasing.
Take $A,B\in \mathcal F$ with  $A\subseteq B$ and $v(A)=1$. We will show $v(B)=1$.
Let $f$ be a continuous increasing function on $\R$ satisfying $f(0)=0$ and $f(1)=f(2)=2$.
Using  (b), we have
$\cR(\id_A+\id_B)=\cR(\id_A)+\cR(\id_B) = v(A) + v(B) \in \{1,2\}$.
Hence, $f(\cR(\id_A+\id_B))=2$.
On the other hand, by (c),
$\cR(f(\id_A+\id_B))=\cR(2 \id_B )=2\cR( \id_B)=2v(B)$. Using ordinality we have $v(B)=1$.
Therefore, $v$ is increasing. This shows that $v$   is a capacity taking values in $\{0,1\}$ and hence $\cR$ is a Choquet quantile on $\X_0$.

Finally, we show that the above representation $\cR(X) = \int X\d v$ also holds on $\X_{\mathcal A}$.
Take any $X\in \X_{\mathcal A}$ and let $X_n= \lfloor nX \rfloor/n$ for each $n\in \N$ where  $ \lfloor x \rfloor$  is the largest integer dominated by $x$.
For each $n$, $X$ and $X_n$ are comonotonic and $\Vert X-X_n\Vert\le 1/n$ where $\Vert \cdot \Vert$ is the supremum norm. Using Denneberg's Lemma, there exist continuous and increasing functions $f_n$ and $g_n$
such that $X=f_n(Z_n)$ and $X_n=g_n(Z_n)$ where $Z_n=X+X_n$.
Note that $|f_n(x)-g_n(x)|\le 1/n$ for $x $ in the closure of  $\{Z_n(\omega),\omega\in\Omega\}$.
Using the above fact together with ordinality and (a), we get
\begin{align*}
| \cR(X)-\cR(X_n)| = |\cR(f_n(Z_n)) - \cR(g_n(Z_n))|   &= |f_n (\cR(Z_n)) - g_n(\cR(Z_n))| \le \frac 1n .
\end{align*}
Moreover,  since Choquet integrals  are $1$-Lipschitz continuous (e.g., Proposition 4.11 of \cite{MM04}), $\Vert X-X_n\Vert\le 1/n$ implies $
| \int X_n \d v -\int X \d v |  \le 1/n $.
Therefore, $|\cR(X)-\int X\d v|\le 2/n$. Sending $n\to\infty$ yields $\cR(X) = \int X\d v$. This completes the proof. 
\qed\medskip

\noindent{\it Proof of Proposition \ref{prop:4}. } Note that (ii) implies (i) is trivial. We next show (i) implies (ii). Let $\mathcal R$ denote one of the numerical representations. Then it follows that
\begin{align}\label{eq:RR}
\mathcal R(X)\leq \mathcal R(Y) \implies \mathcal R(\phi(X))\leq \mathcal R(\phi(Y))
\end{align} for all increasing and continuous function $\phi$. As the preference is increasing, we have $\mathcal R(0)<\mathcal R(1)$. Let $h(x)=\mathcal R(x),~x\in\R$, $\mathcal R(\mathcal X_{\mathcal A})=\{\mathcal R(X): X\in\mathcal X_{\mathcal A}\}$, and $\mathcal R(\R)=\{\mathcal R(x): x\in\R \}$. Following  the argument in the proof of Lemma 2 of \cite{FLW22}, we have $h$ is strictly increasing and $\mathcal R(\mathcal X_{\mathcal A})=\mathcal R(\R)$. Define $\mathcal R'(X) = h^{-1}(\mathcal R(X))$ for $X\in \X_{\mathcal A}$, where $h^{-1}(x)=\inf\{y: h(y)\geq x\}$ is continuous and increasing over $(h(-\infty),h(\infty))$.  We can verify that $ \mathcal R'$ satisfies
$\mathcal R'(c)= h^{-1} (\mathcal R(c)) = h^{-1} (h(c))=c ~\text{for all}~ c\in \R.$ Hence, 
 $\mathcal R'(\mathcal R'(X))=\mathcal R'(X) $ for $X\in\X_{\mathcal A}$. By \eqref{eq:RR}, we have  $\mathcal R(X)=\mathcal R(Y)$ implies $\mathcal R(\phi(X))= \mathcal R(\phi(Y))$ and therefore,
 $\mathcal R'(X)=\mathcal R'(Y)$ implies $\mathcal R'(\phi(X))= \mathcal R'(\phi(Y))$.
Using the fact  $\mathcal R'(\mathcal R'(X))=\mathcal R'(X) $ for any $X\in \X_{\mathcal A}$, we have 
 $\mathcal R'(\phi(\mathcal R'(X)))=\mathcal R'(\phi(X))$. By  $\mathcal R'(c)= c $, $c\in\R$, the above equation  can be rewritten as
 $\phi(\mathcal R'(X)) =\mathcal R'(\phi(X)) $ for all increasing and continuous $\phi$. Consequently, $\mathcal R'$ is the desired numerical representation.
\qed\medskip

\noindent{\it Proof of Proposition \ref{pro:supclose}.} 
 Recall that $\overline w=\sup_{w\in W} w$ and $\underline{w}=\inf_{w\in W} w$. 
It follows that 
|
 \begin{align*}
 \sup_{w\in W} Q^w_\alpha(X) &=   \sup_{w\in W} \inf \{ x\in \R : w(X\ge x) \le 1-\alpha\}
 \\ &=  \inf \bigcap_{w\in W} \{ x\in \R : w(X\ge x) \le 1-\alpha\}
 \\ &=    \inf \{ x\in \R : w(X\ge x) \le 1-\alpha \mbox{~for all $w\in W$} \}
 \\&=    \inf \left\{ x\in \R : \sup_{w\in W} w(X\ge x) \le 1-\alpha\right\} =  Q^{\overline w}_\alpha(X).
 \end{align*}
Moreover,  we have that 
\begin{align*}
\inf_{w\in W} \overline{Q}^w_\alpha(X) &=   \inf_{w\in W} \inf \{ x\in \R : w(X\ge x)<1-\alpha\}
\\ &=  \inf \bigcup_{w\in W} \{ x\in \R : w(X\ge x)< 1-\alpha\}
\\ &=    \inf \{ x\in \R : w(X\ge x)< 1-\alpha \mbox{~for some $w\in W$} \}
\\&=    \inf \left\{ x\in \R : \inf_{w\in W} w(X\ge x)<1-\alpha\right\} =  \overline{Q}^{\underline{w}}_\alpha(X).
\end{align*}
The other statement on $v$ is similar. 
This completes the proof.  \qed\medskip

\noindent{\it Proof of Proposition \ref{prop:7}.}
It follows from Theorem 4.94 of \cite{FS16} that $I_v$ is convex if and only if $v$ is submodular.  As $v$ only takes values $0$ and $1$, we have $v(A\cup B)+v(A\cap B)\leq v(A)+v(B)$ is equivalent to the closeness of $\mathcal C_v$ under union.
\qed\medskip

\noindent{\it Proof of Proposition \ref{prop:8}.}  
The ``if" part is obvious. Next, we show the ``only if" part. 
 By Theorem 4.94 of \cite{FS16},  the convexity of $I_v$ implies that $v$ is submodular.
If $v(A)=1$, then for any $B\subseteq A$, using the fact that $v$ is submodular, we have $v(A)\leq v(B)+v(A\setminus B)$. This implies  either $v(B)=1$ or $v(A\setminus B)=1$. For the three spaces,  we can continue this process to generate a decreasing sequence $\{A_n,~n\in\mathbb N\}$ such that $v(A_n)=1$ and $\bigcap_{n=1}^\infty A_n=\varnothing$ or $\{x\}$ for some $x\in A$. Note that if $\bigcap_{n=1}^\infty A_n=\varnothing$, then the continuity of $v$ implies $v(A_n)\to 0$, which contradicts the fact that $v(A_n)=1$. Hence, $\bigcap_{n=1}^\infty A_n=\{x\}$ for some $x\in A$. It follows from Theorem 4.2 of \cite{MM04} and the continuity of $v$ that $v(A_n)\to v(A)$ as $n\to\infty$ if $A_n\downarrow A$. Hence, we have $v(\{x\})=\lim_{n\to\infty} v(A_n)=1$. Let $D=\{x\in \Omega: v(\{x\})=1\}$. Then $D$ consists of finite elements. If not, for distinct $x_n\in D$ with $n\geq 1$, let $A_n=\{x_m: m\geq n\}$. Using the fact that $A_n \downarrow \varnothing$ and the continuity of $v$, we have $v(A_n)\to 0$ as $n\to\infty$, leading to a contradiction. Hence, $D$ consists of finite elements. Moreover, $D\in \mathcal F$ and $v(D^c)=0$. This implies $I_v(X)=\max_{x\in D} X(x)$. We complete the proof.
\qed\medskip

\noindent{\it Proof of Proposition \ref{prop:2}.}
First note that  $Q^w_\alpha(X)\in\R$ as $X$ is bounded.  By definition, we have for any $\epsilon>0$, $ Q^w_\alpha(X)-\epsilon \not\in \{x:  w (X\ge x) \le 1-\alpha\}$, and thus, 
$s:=w(X\geq Q^w_\alpha(X)-\epsilon)>1-\alpha$. Then for any $\beta \in (1- s, \alpha)$, it holds that $w(X\geq Q^w_\alpha(X)-\epsilon)=s>1-\beta$, implying $Q^w_\beta(X)\geq Q^w_\alpha(X)-\epsilon$. Note that $Q^w_\alpha(X)$ is increasing in $\alpha\in (0,1]$. Hence $Q^w_\beta(X)\leq Q^w_\alpha(X)$. Thus we have $|Q^w_\beta(X)-Q^w_\alpha(X)|\leq \epsilon$. This implies that $Q^w_\alpha(X)$ is left-continuous on $(0,1]$. One can similarly show the right continuity of $\overline{Q}^w_\alpha(X)$ on $[0,1)$. The details here are omitted.
\qed\medskip

\noindent{\it Proof of Theorem \ref{thm:inf}. } (i) 
First suppose $\mathcal C_v\ne \mathcal F$. 
 We start with the case  $X\ge 0$.  For $i\in [n]$,  take $A_i\in \mathcal C_{v_i}$  and let  $X_i=X\id_{A_i}$ and $B=\bigcup_{i=1}^n A_i \in \mathcal C_v$. Note that the monotonicity of $I_{v_i}$ implies that for all  $Z, Y\in\X $, 
 $$I_{v_i}(Z)\le I_{v_i}(Z-Y+\sup Y)=I_{v_i}(Z-Y)+ \sup Y,~~i\in [n].$$  Hence, we have
 $$ \dsquare_{i=1}^n I_{v_i} (X)
 \le \sum_{i=1}^n I_{v_i} (X_i) + \sup\left ( X -\sum_{i=1}^n X \id_{A_i}\right )
\le \sum_{i=1}^n I_{v_i} (X_i) + \sup(X   \id_{B^c}).$$
Direct computation gives 
 $I_{v_i} (X_i)  = 0$ for all $ i\in [n]$, which  implies $$\dsquare_{i=1}^n I_{v_i} (X)\le \sup(X\id_{B^c})=\inf\{x\in\R: \{X>x\}\subseteq B\}.$$
 Consequently, by the arbitrariness of $B\in \mathcal C_v $, we have
\begin{align*}
\dsquare_{i=1}^n I_{v_i} (X)
  \le 
  \inf_{B\in \mathcal C_v} \inf\{x\in\R: \{X>x\}\subseteq B\}=\inf\{x\in\R: \{X>x\}\in \mathcal C_v\} = I_v(X). 
 \end{align*}
  The above inequality holds for general $X\in \X$ by   the translation invariance of $I_v. $
 
 Next, we show the inverse inequality. For $(X_1,\dots,X_n)\in \mathbb A_n(X)$, let $c=\sum_{i=1}^n I_{v_i} (X_i)$ and $c_i=I_{v_i} (X_i)+\frac{\epsilon}{n}$ for some $\epsilon>0$.  Note that $I_{v_i}(X_i-c_i)<0$ for all $i\in [n]$. This implies
 $\{X_i\geq c_i\}\in \mathcal C_{v_i}$  for all $i\in [n]$. Hence, we have
 $\{X\geq c+\epsilon\}\subseteq \bigcup_{i=1}^n\{X_i\geq c_i\}\in \mathcal C_{v}$. This implies $\{X\geq c+\epsilon\}\in \mathcal C_{v}$ and
 $I_v(X)\le c+\epsilon$.  Letting $\epsilon\to 0$, we have $I_v(X)\le \sum_{i=1}^n I_{v_i} (X_i)$, implying $I_v(X)\le \dsquare_{i=1}^n I_{v_i} (X)$. Hence, the desired equality holds.
 
  Finally, we consider the case $\mathcal C_v= \mathcal F$. Note that there exist disjoint $A_i\in \mathcal C_{v_i}$, $i\in [n],$ such that $\bigcup_{i=1}^nA_i=\Omega$. Letting  $X_i=(X+c)\id_{A_i}$, $i\in [n],$ for some $c>0$ such that $X+c\geq 0$, we have $\dsquare_{i=1}^n I_{v_i} (X+c)
 \le \sum_{i=1}^n I_{v_i} (X_i)=0$, implying $\dsquare_{i=1}^n I_{v_i} (X)\leq -c$. Hence, $\dsquare_{i=1}^n I_{v_i} (X)=-\infty$ by letting $c\to \infty$.
 
 (ii) Direct computation gives $I_{v_i}(X_i^*)\leq {x^*}/{n}$ for $i\in [n]$. Hence, $\sum_{i=1}^n I_{v_i}(X_i^*)\leq x^*$, implying that $(X_1^*,\dots, X_n^*)$ is an optimal allocation. This completes the proof.
 \qed\medskip

\noindent{\it Proof of Proposition} \ref{prop:invariance}. The equation follows from 
 $$
\dsquare_{i=1}^n I_{v_i} ( \phi(X))=I_v (\phi (X))=\phi( I_v (X))=\phi\left(\dsquare_{i=1}^n I_{v_i} ( X)\right),
$$
where the first and last equalities follow from (i) of Theorem \ref{thm:inf} and the second equality follows from the oridinality of Choquet quantiles by Theorem \ref{th:main}. The optimal allocation follows immediately by calculating $\sum_{i=1}^n I_{v_i} ( \phi(Y^*_i))$ which achieves the optimal value. \qed\medskip

  \noindent{\it Proof of Corollary \ref{co:6}.} Let $\mathcal C_{v_i}=\{A\in\mathcal F: \sup_{\mathbb Q\in \mathcal P_i} \mathbb Q(A)\leq \alpha_i\}$ for $i\in[n]$. Then it follows from Proposition \ref{pro:supclose} that $\sup_{\mathbb Q\in \mathcal P_i} \VaR^{\mathbb Q}_{\alpha_i}=I_{v_i},~i\in [n]$. Applying Theorem \ref{thm:inf}, we immediately obtain the desired conclusion.
 \qed\medskip

  \noindent{\it Proof of Corollary \ref{co:4}.} Applying Corollary 2 of \cite{ELW18}, we have
     \begin{align*}
    \dsquare_{i=1}^n\VaR_{\alpha_i}^{g_i\circ \mathbb Q}(X)=\dsquare_{i=1}^n\VaR_{g_i^{-1}(\alpha_i)}^{\mathbb Q}(X)=\VaR_{\sum_{i=1}^n g_i^{-1}(\alpha_i)}^{\mathbb Q}(X).
 \end{align*}
This completes the proof. \qed\medskip

\noindent{\it Proof of Proposition} \ref{prop:9}. We only give the proof for the case of atomless $(\Omega,\mathcal F,\p) $ and the second case can be shown similarly.  
  Denote by $\mathcal Z_i=\{{\rm d}\mathbb Q/{\rm d}\mathbb \p:  \mathbb Q\in \mathcal P_i\}$, $i\in [n]$. In light of Lemma 4.89 of \cite{FS16}, there exists $U\sim {\rm U}[0,1]$ under $\mathbb P$ such that $U$ and  $X, -Z_1, Z_2$ are all comonotonic for all $Z_1\in\mathcal Z_1$ and $Z_2\in \bigcup_{i=2}^n\mathcal Z_i$. 
 For any disjoint $A_1,\ldots,A_n$ satisfying 
$  \bigcup_{i=1}^n A_i  = \{X>x^*\}$ and 
$$
\sup_{\mathbb Q\in \mathcal P_i} \mathbb Q(A_i)  = \sup_{Z\in \mathcal Z_i} \int_{A_i}Z\d \p \le \alpha_i,~~i\in [n],
$$
 let $A_1'=\{1-\mathbb P(A_1)<U<1\}$ and $A_i'=(A_i\setminus A_1')\cup B_i,~i=2,\dots, n$, where $B_2,\dots, B_n$ are disjoint satisfying $\bigcup_{i=2}^n B_i=A_1\setminus A_1'$ and $\mathbb P(B_i)=\mathbb P(A_i\cap A_1')$.  Then we have $\bigcup_{i=1}^n A_i'=\bigcup_{i=1}^n A_i=\{X>x^*\}$. By the construction of $A_1'$ and  the countermonotonicity of $U$ 
 and $Z\in\mathcal Z_1$, we have $\int_{A_1'}Z\d \p \leq \int_{A_1}Z\d \p$ for any $Z\in\mathcal Z_1$. Moreover, it follows from the definition of $A_i'$ and the comonotonicity of $U$ and $Z \in\mathcal Z_2$ that $$\int_{A_i'}Z\d \p=\int_{A_i\setminus A_1'}Z\d \p+\int_{B_i}Z\d \p\leq \int_{A_i\setminus A_1'}Z\d \p+\int_{A_i\cap A_1'}Z\d \p=\int_{A_i}Z\d \p$$ for all $Z\in \bigcup_{i=2}^n\mathcal Z_i$.
 Hence, we have
 $$
 \sup_{Z\in \mathcal Z_i} \int_{A_i'}Z\d \p \leq \sup_{Z\in \mathcal Z_i} \int_{A_i}Z\d \p \le \alpha_i,~~i\in [n].
$$
 Let $A_1^*=A_1', \dots, A_{n-1}^*=A_{n-1}'$ and $A_n^*=\left(\bigcup_{i=1}^{n-1} A_i^*\right)^c$. Then $(A_1^*,\dots, A_n^*)\in \Pi_n$ is the optimal composition satisfying $\id_{A_1^*}-\id_{\bigcup_{i=2}^m A_i^*} $ and $X$ are comonotonic. This completes the proof.
 \qed\medskip

\noindent{\it Proof of Proposition \ref{prop:ESC}.} 
 It follows from Theorem 4.94 of \cite{FS16} that $I_{w_{\alpha }}$ is a coherent risk measure if and only if $w_{\alpha }$ is submodular.
\qed\medskip

\noindent{\it Proof of Theorem \ref{th:min}.}
Let $g(x)=x+\frac{1}{ \alpha}\int(X-x)_+\d w=x+ \frac{1}{ \alpha}\int_{x}^{\infty} w(X>y)\d y$.
Note that $g$ is a continuous and convex function over $\R$ and $\lim_{x\to\pm\infty}g(x)=\infty$. Hence the minimizers of $g$ exist as a finite interval or singleton. Let $h^+$ ($h^-$) be the right (left)-continuous version of the function  $x\mapsto w(X>x)$. Then the minimizers of $g$ satisfy  $g_+'(x)=1-\frac{h^+(x)}{\alpha}\geq 0$ and $g_-'(x)=1-\frac{h^-(x)}{\alpha}\leq 0$, which is equivalent to $h^+(x)\leq \alpha\leq h^-(x)$. Note that $h^+(x)\leq \alpha$ is equivalent to $x\geq \VaR_{\alpha}^w(X)$ and $h^-(x)\geq \alpha$ is equivalent to $x\leq \overline{\VaR}_{\alpha}^w(X)$. Hence we obtain the formula for Choquet quantiles.

 We next derive the formula for Choquet ES. For $x=\VaR_{\alpha}^w(X)$, it follows that
\begin{align*}
g(x)&=x+ \frac{1}{ \alpha} \int_{x}^{\infty} \int_{0}^{1} \id_{\{w(X\ge y)>t\}}\d t\d y \\
&=x+ \frac{1}{ \alpha}  \int_{0}^{1}\int_{x}^{\infty} \id_{\{w(X\ge y)>t\}}\d y \d t\\
&=x+ \frac{1}{ \alpha} \int_{0}^{1}\int_{x}^{\infty} \id_{\{y<\VaR_{t}^w(X)\}}\d y\d t\\
&=x+ \frac{1}{ \alpha} \int_{0}^{\alpha}   (\VaR_{t}^w(X) -x)_+ \d t =\frac{1}{ \alpha}\int_{0}^{\alpha}  \VaR_{t}^w(X)\d t =\ES^w_\alpha(X),
\end{align*}
where the third equality follows from the fact that $w(X\ge y)>t$ implies that $y\le \VaR_{t}^w(X)$,  $y<\VaR_{t}^w(X)$ implies that $w(X\ge y)>t$, and $\int_{x}^{\infty} \id_{\{y<\VaR_{t}^w(X)\}}\d y = \int_{x}^{\infty} \id_{\{y\le \VaR_{t}^w(X)\}}\d y$.
We complete the proof.
\qed\medskip

\noindent{\it Proof of Proposition \ref{prop:ESS}.}
We first note that for any submodular capacity $u$, the mapping $I_{u}$ admits a representation as the supremum of $\E^{\mathbb Q}$ over all finitely additive capacities $\mathbb Q\le u$; see Theorem 4.94 of \cite{FS16}. Therefore, it suffices to show that these finitely additive capacities can be chosen as countably additive. This boils down to checking continuity in the form $I_{u}(X_n)\downarrow I_{u}(X)$ for $X_n\downarrow X$ pointwise; see e.g., Theorem 4.22 of \cite{FS16}. We check this below, where $u=w_\alpha$ so that by Proposition \ref{prop:ESC}, $\ES_\alpha^w=I_{u}$.

Note that the continuity of $w$ implies the continuity of $w_{\alpha}$.
By   Theorem 4.2 of \cite{MM04}, this implies $\lim_{n\to\infty} w_{\alpha}(A_n)=w_{\alpha}(A)$ if $A_n\downarrow A$.
For $X_n\downarrow X$ pointwise with $X\geq 0$, we have $\{X_n\geq x\}\downarrow \{X\geq x\}$ for all $x\in\R$. Using the above conclusion and the monotone convergence theorem,  we have
\begin{align*}\lim_{n\to\infty}\ES_\alpha^w(X_n)&=\lim_{n\to\infty}\int_{0}^{\infty} w_{\alpha}(X_n\geq x) \d x=\int_{0}^{\infty} \lim_{n\to\infty}w_{\alpha}(X_n\geq x) \d x\\
&=\int_{0}^{\infty}w_{\alpha}(X\geq x) \d x=\ES_\alpha^w(X).
\end{align*}
The above conclusion also holds for general $X\in\X$ by using the fact that $\ES_\alpha^w(X+c)=\ES_\alpha^w(X)+c$ for $c\in\R$.
\qed\medskip

\noindent{\it Proof of Proposition \ref{prop:coherentchoquet}.}  Suppose there exists $w\in \mathcal S$ such that $w\leq \bigwedge_{i=1}^nw_i$. Using the monotonicity of the Choquet integral with respect to the capacity $w$ and the subadditivity of $I_{w}$, we have that for all $(X_1,\dots, X_n)\in \mathbb{A}_n(X)$, 
$$\sum_{i=1}^{n} I_{w_i}(X_i)\geq \sum_{i=1}^{n} I_{w}(X_i)\geq  I_{w}(X).$$
 Hence, $\dsquare_{i=1}^n I_{w_i}(X)\geq I_{w}(X)$, which further implies $\dsquare_{i=1}^n I_{w_i}\geq \sup_{w\in \mathcal S,~w\leq \bigwedge_{i=1}^nw_i}I_{w}>-\infty$. We next consider the other direction. Applying Theorem 3.6 of \cite{BE05}, we have \begin{align}\label{inf:choquet}\dsquare_{i=1}^n I_{w_i}(X)=\sup_{{\mathbb Q}\in \mathcal P^*; 
  {\mathbb Q}\leq \bigwedge_{i=1}^nw_i} \E^{\mathbb Q}(X) \leq \sup_{w\in \mathcal S; w\leq \bigwedge_{i=1}^nw_i}I_{w}(X),
\end{align}
where $\mathcal P^*$ is the collection of all finitely additive capacities on $(\Omega,\mathcal F)$ and the   inequality holds due to $\mathcal P^*\subseteq \mathcal S$. 

Note that the equality \eqref{inf:choquet} together with Theorem 3.6 of \cite{BE05} implies that if $\dsquare_{i=1}^n I_{w_i}(0)>-\infty$, then there exists ${\mathbb Q}\in \mathcal P_1$ such that ${\mathbb Q}\leq \bigwedge_{i=1}^nw_i$.
Hence, if there is no $w\in \mathcal S$ such that $w\leq \bigwedge_{i=1}^nw_i$, then $\dsquare_{i=1}^n I_{w_i}(0)=-\infty$, which implies $\dsquare_{i=1}^n I_{w_i}(X)=-\infty$ for all $X\in\X$. 
\qed\medskip

\noindent{\it Proof of Corollary \ref{coro:ES-inf}.}
This follows from Propositions \ref{prop:ESC} and \ref{prop:coherentchoquet}. 
\qed\medskip



\section{Additional numerical details and figures}
\label{app:numerical}

The probabilities in Examples \ref{ex:3}--\ref{ex:4-1}  are defined as 
$$
\mathcal P_i =\left\{ \mathbb Q_{k}^{(i)}: \frac{{\rm d} \mathbb Q_{k}^{(i)}}{{\rm d} \mathbb \p}=\sum_{j=1}^{20} p_{kj}^{(i)} \id_{B_j},~ k\in [3]  \right\},~~i\in [5],
$$
 where $p_{k}^i = (p_{k1}^i,\ldots,p_{km}^i)$ for $k\in [3]$, $i\in [5]$ and $m=20$ are randomly generated (uniformly on the simplex) using \texttt{Python} with seed number 42. We report simulated probabilities $p_{k}^i$ for $k\in [3]$ and $i\in [5]$  in Figure \ref{fig:pvalues} for  reproduction. 

 \begin{figure}[t] 
    \centering
    \includegraphics[width=6cm]{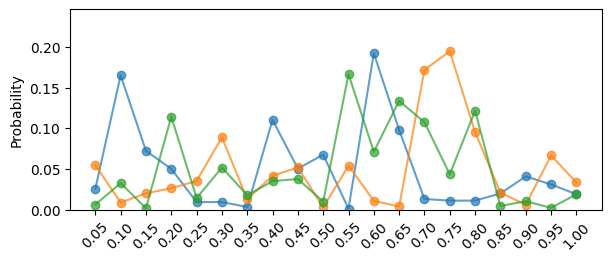} 
    \includegraphics[width=6cm]{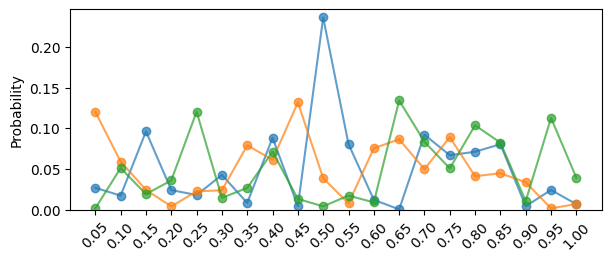} 
     \includegraphics[width=6cm]{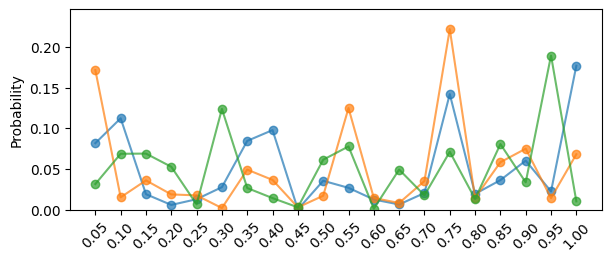} 
    \includegraphics[width=6cm]{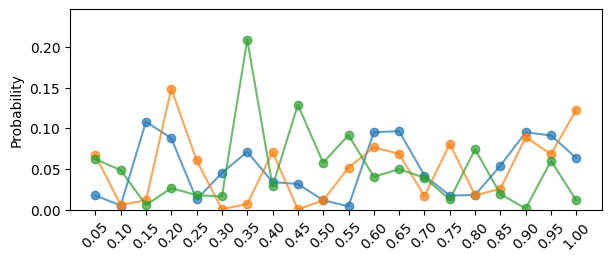}
     \includegraphics[width=6cm]{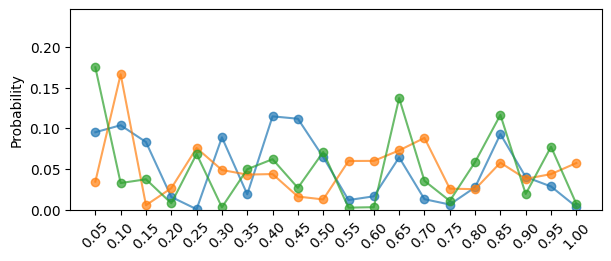}  
    \caption{The  probability vectors for probabilities in $\mathcal P_i$ for $i\in [5]$ in Section \ref{sec:numerical}}\label{fig:pvalues}
\end{figure}

 The optimal allocation among five agents in the atomless setting in Example \ref{ex:3} and in the discrete setting in Example \ref{ex:4-1} are reported in Figures \ref{fig:22} and \ref{fig:2}, respectively.
   \begin{figure}[t]
    \centering      \begin{subfigure}[b]{0.99\textwidth}    \centering
    \includegraphics[width=12.5cm]{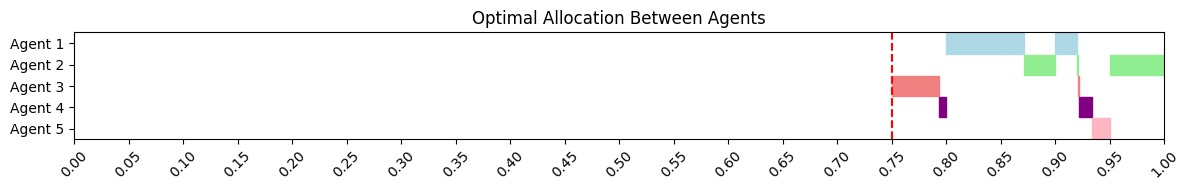}   \caption{$(\alpha_1,\dots,\alpha_5) = (0.1, 0.1, 0.05, 0.05, 0.05)$} 
     \label{fig:21}
    \end{subfigure} \\~\\
         \begin{subfigure}[b]{0.99\textwidth}     
    \centering
    \includegraphics[width=12.5cm]{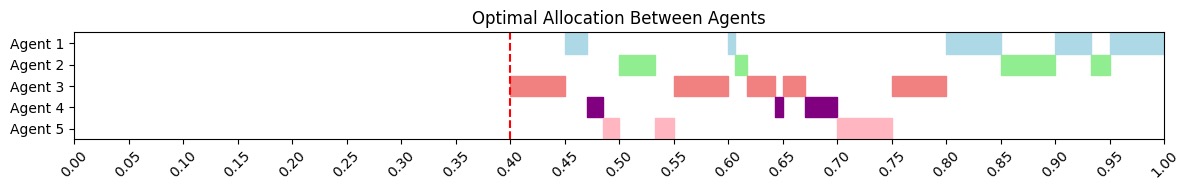}
  \caption{$(\alpha_1,\dots,\alpha_5) = (0.05, 0.05, 0.025, 0.025, 0.025)$}\label{fig:31-atomless}
    \end{subfigure}
    \caption{Optimal allocation among five agents in the atomless setting in Example \ref{ex:3}}
    \label{fig:22}
\end{figure}

  \begin{figure}[t]
    \centering
    \includegraphics[width=12.5cm]{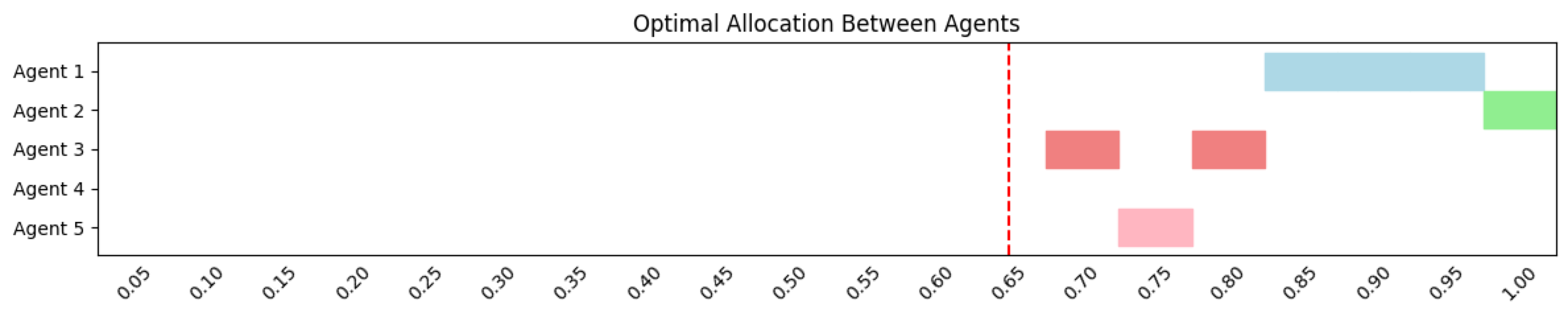} 
    \caption{Optimal allocation among five agents in the discrete setting in Example \ref{ex:4-1}, where $(\alpha_1,\dots,\alpha_5) = (0.1, 0.1, 0.05, 0.05, 0.05)$}
    \label{fig:2}
\end{figure}

The 10 distributions in ambiguity set  $\mathcal P_3$ and the 20 distributions in ambiguity set  $\mathcal P_4$ in Section \ref{sec:realdata} are given in Figures \ref{fig:RD-SAA-1year} and \ref{fig:RD-normalt-all}, respectively.
\begin{figure}[t] 
    \centering
    \includegraphics[width=12cm]{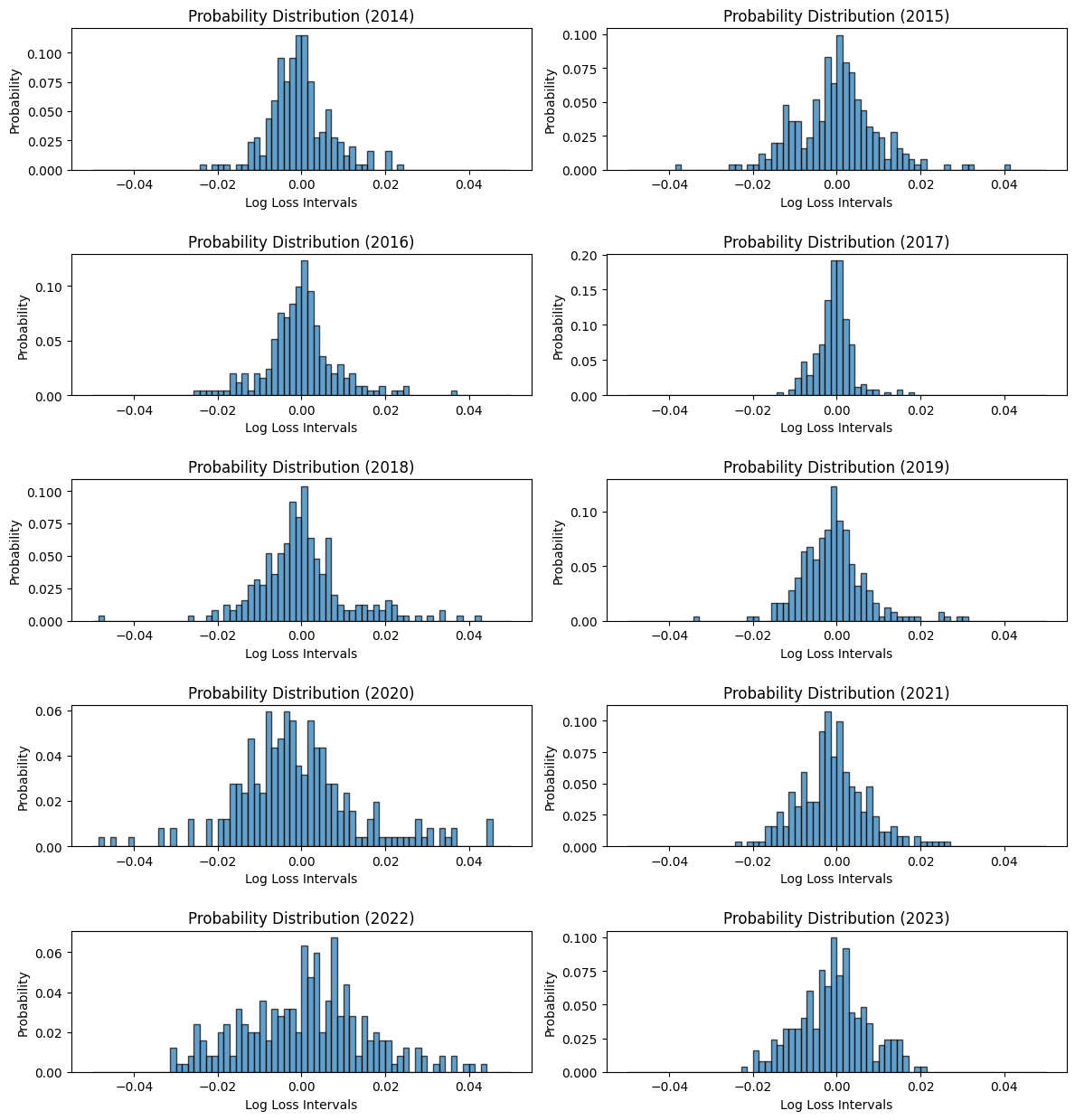} 
    \caption{10 probabilities  for the log loss of agent $3$}\label{fig:RD-SAA-1year}
\end{figure}  

\begin{figure}[t] 
    \centering
    \includegraphics[width=12cm]{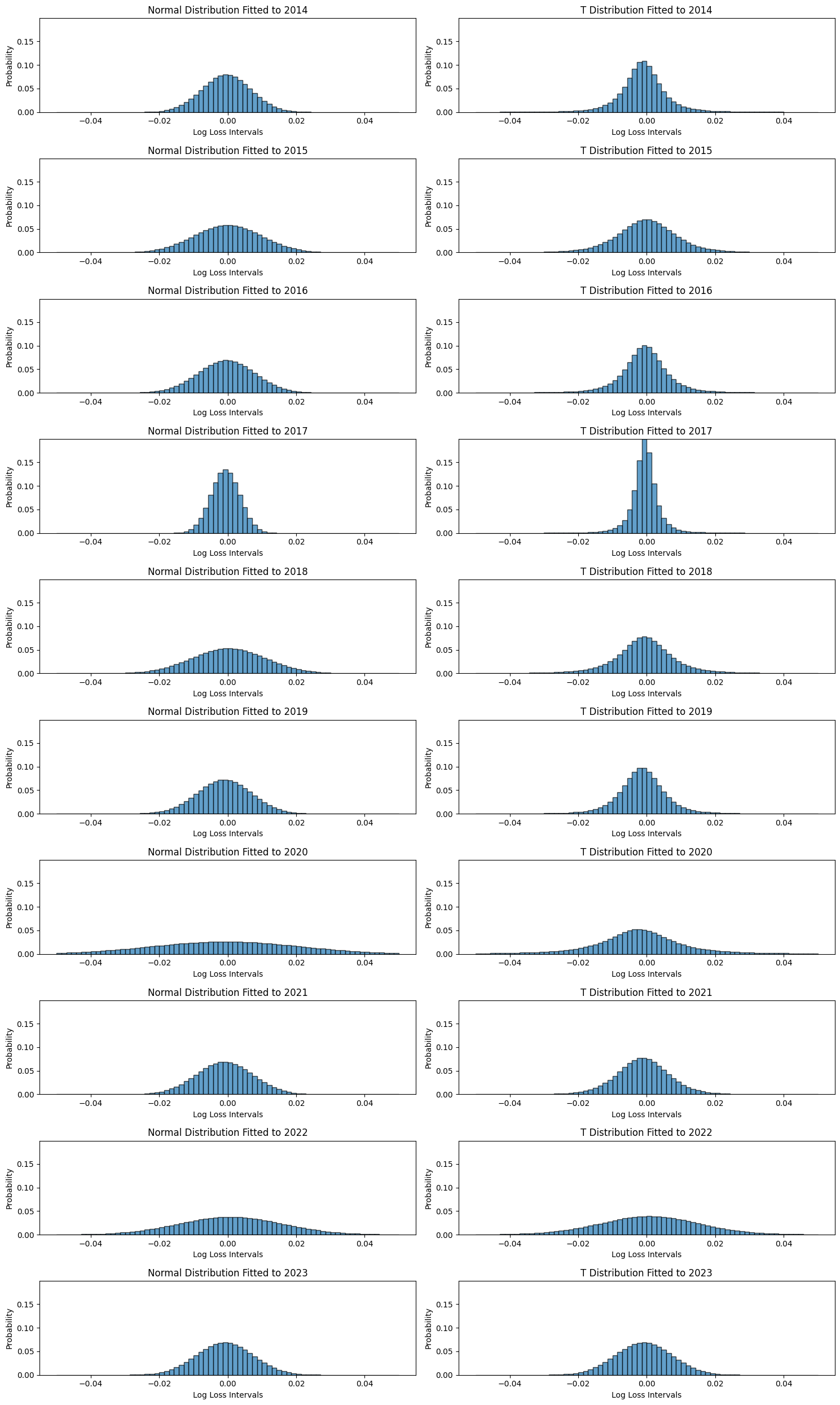} 
    \caption{20 probabilities  for the log loss of agent $4$ fitted to normal and t distributions}\label{fig:RD-normalt-all}
\end{figure}

\end{document}